\documentclass[11pt]{article}
\usepackage{amsmath}
\usepackage{amsthm}
\usepackage{thmtools, thm-restate}
\usepackage{amssymb}
\usepackage{fullpage}
\usepackage{algorithmicx,algpseudocode}
\usepackage{algorithm}
\usepackage{xcolor}
\usepackage[english]{babel}
\usepackage{graphicx}
\graphicspath{{images/}{../images/}}
\usepackage{subfiles}
\usepackage{blindtext}
\usepackage{enumitem}
\usepackage[numbers]{natbib}

\newcommand{\ALG}{\text{ALG}}
\newcommand{\OPT}{\text{OPT}}
\newcommand{\Ex}{\mathbb{E}}
\newtheorem{theorem}{Theorem}
\newtheorem{lemma}{Lemma}
\newtheorem{definition}{Definition}
\newtheorem{prop}{Proposition}

\newcommand{\nameOfProblem}{submodular $k$-secretary problem}
\newcommand{\nameOfProblemSL}{submodular $k$-secretary problem with shortlists}
\newcommand{\nameOfProblemMatroid}{submodular matroid secretary problem}
\newcommand{\nameOfProblemMatroidSL}{submodular matroid secretary problem with shortlists}

\newcommand{\streamingProblem}{submodular random order streaming problem}

\newcommand{\nameSecSL}{secretary problem with shortlists}

\newcommand{\scomment}[1]{}
\newcommand{\ccomment}[1]{}
\newcommand{\mcomment}[1]{}
\newcommand{\toRemove}[1]{}

\newcommand{\etaMacro}{$c\frac{\log(1/\epsilon)}{\epsilon^2}  {\frac{1}{\epsilon^6} \log(1/\epsilon)  \choose {\frac{1}{\epsilon^4} \log(1/\epsilon)}}$~}
\usepackage{mathtools}

\title{Submodular Matroid Secretary Problem with Shortlists}

\date{ }

\begin{document}

\author{Mohammad Shadravan \thanks{Columbia University, \texttt{ms4961@columbia.edu}} 
}

\toRemove{
\author{Shipra Agrawal  \thanks{Columbia University, \texttt{sa3305@columbia.edu}. Research supported in part by Google Faculty Research Awards 2017 and Amazon Research Awards 2017.
}	
\and 	Mohammad Shadravan \thanks{Columbia University, \texttt{ms4961@columbia.edu}} 	
\and 	Cliff Stein \thanks{Columbia University, \texttt{cliff@ieor.columbia.edu}. Research supported in part by NSF grants CCF-1421161 and CCF-1714818.} }
}





\maketitle


\begin{abstract} 
In the \textit{matroid secretary problem}, which is a generalization of the classic secretary problem,
the elements of a matroid $\mathcal{M}$ arrive in random order. Once we observe an item we need to irrevocably decide 
whether or not to accept it. The set of selected elements should form an independent set of the matroid.
The goal is to maximize the total sum of the values assigned to these elements. The existence of a constant competitive algorithm
is a long standing open problem.

In this paper, we introduce a version of this problem, which we refer to as submodular matroid secretary problem with shortlists (motivated by the \textit{shortlist} model in~\cite{us}). 
In this setting, the algorithm is allowed to choose a subset of items as part of a shortlist, possibly more than $k=rk(\mathcal{M})$ items. Then, after seeing the entire input, the algorithm can choose an independent subset from the shortlist. Furthermore we generalize the objective function to any monotone submodular function.
The main question is that can an online algorithm achieve a constant competitive ratio using a shortlist of size $O(k)$?

We design an algorithm that achieves a $\frac{1}{2}(1-1/e^2-\epsilon-O(1/k))$ competitive ratio for any constant $\epsilon>0$, using a shortlist of size $O(k)$. This is especially surprising considering that the best known competitive ratio for the matroid secretary problem is $O(\log \log k)$. We are also able to get a constant competitive algorithm using shortlist of size  at most $k$ and also a constant competitive algorithm in the preemption model.

An important application of our algorithm is for the random order streaming of submodular functions.
We show that our algorithm can be implemented in the streaming setting using $O(k)$ memory. It achieves a $\frac{1}{2}(1-1/e^2-\epsilon-O(1/k))$ approximation. 
The previously best known approximation ratio for streaming submodular maximization under matroid constraint  is 0.25 (adversarial order) due to ~\citet{feldman2018streaming},~\citet{Chekuri} and ~\citet{Chakrabarti2015}.
Moreover, we generalize our results to the case of $p$-matchoid constraints and give a $\frac{1}{p+1}(1-1/e^{p+1}-\epsilon-O(1/k))$ approximation using $O(k)$ memory,
which asymptotically (as $p$ and $k$ increase) approaches the best known offline guarantee $\frac{1}{p+1}$~\cite{nemhauser1978analysis}. 

\end{abstract}


\section{Introduction}

In recent years, submodular optimization has found
 applications 
for different 
machine learning and data mining  applications including data summarization,
  sparsity,  active learning, recommendation,  high-order graphical model inference, determinantal point processes~\cite{feldman2018streaming, Badanidiyuru2014StreamingFly,kazemi2019submodular}, network inference, network design~\cite{ghugequasi,friggstad2014linear},  and influence
maximization in social networks~\cite{kazemi2019submodular}.

In these applications,  the data is generated in a real time, 
and it is important to keep track of the data that is seen so far. Consequently, a line of recent papers studied streaming algorithms for maximizing a submodular function.
The first  one-pass streaming algorithm for maximizing a monotone submodular function subject to a $k$-cardinality constraint is due to ~\citet{Badanidiyuru2014StreamingFly}, who propose a ($1/2-\epsilon$)-approximation streaming algorithm, 
with a memory of size $O(\frac{1}{\epsilon}k\log k)$.
Recently,~\citet{kazemi2019submodular} proposed a new algorithm with the same approximation ratio but with improved memory $O(k)$. 

\citet{norouzi} give an upper bound of $1/2+o(1)$ on the approximation ratio achievable by
any algorithm for streaming submodular
maximization that only queries the value of the submodular function on feasible sets (sets of size at most $k$) while using $o(n)$ memory.
Consequently, they initiate the study of the random order streaming model in order to go beyond this worst case analysis for the adversarial order inputs. They achieve a $1/2+8\times 10^{-14}$ approximation for maximizing a monotone submodular function in the random order model, using a memory buffer of size $O(k\log k)$.
Subsequently, \citet{us} substantially improve their result to  $1-1/e-\epsilon-O(1/k)$ approximation which is close to the best possible guarantee in the offline setting, i.e., $1-1/e$ (assuming $P\neq NP$). Furthermore, they improve the required memory buffer to only $O(k)$. 

In addition to the simple cardinality constraint, more general constraints
have been studied in the literature.
~\citet{Chakrabarti2015} give a $1/4p$ approximation algorithm for streaming monotone submodular functions maximization subject to the intersection of $p$ matroid constraints.  
\citet{Chekuri}
extend it to $p$-matchoids constraints.  A precise definition of
a $p$-matchoid is given in Section~\ref{sec:matchoid}. 
These constraints are generalization of   constraints such as the cardinality constraint,
the intersection of $p$ matroids, and matchings in graphs and hyper-graphs~\cite{Chekuri}.
Recently,~\citet{feldman2018streaming} designed a more efficient algorithm with lower number of function evaluations achieving  the same approximation $1/4p$.

The algorithms
of~\citet{feldman2018streaming},
for monotone submodular objective
functions require only $O(k)$ memory ($k$ is the size of the largest feasible
set) and using only $O(kq)$ value and independence oracle queries for processing  each element of the stream ($q$ is a the number of matroids used to define the $p$-matchoid constraint).

Moreover, the greedy algorithm achieves a  $1/(p + 1)$ approximation for $p$-independent systems~\cite{nemhauser1978analysis},  which is tight for all $p$, even for the case of $p$-matroid constraints. Also it is $NP$-hard to approximate the $p$-dimensional matching
within an $\Omega(\log p/p)$ approximation for large but fixed $p$~\cite{hazan2006complexity}.


\toRemove{
\paragraph{Random-order streaming model.}
The {\it  \streamingProblem}~studied in~\cite{norouzi}. In this problem, items from a set $\cal U$ arrive online in random order and the algorithm aims to select a subset $S \subseteq {\cal U}, |S|\le k$ in order to maximize $f(S)$. The streaming algorithm is allowed to maintain a {\it buffer} of size $\eta(k)\ge k$. 

\citet{Hess} initiated the study of  \streamingProblem.
Their algorithm uses $O(k)$ memory and a total of $n$ function evaluations to achieve $0.19$ approximation. 
The state of the art result in the random order input model is due to \citet{norouzi} who achieve a $1/2+8\times 10^{-14}$ approximation, while using a memory buffer of size $O(k\log k)$.
They also give an upper bound of $1/2+o(1)$ on the competitive ratio achievable by  any algorithm for streaming submodular
maximization that only queries the value of the submodular function on feasible sets while using $o(n)$ memory~\cite{norouzi}.
}
\paragraph{\textbf{The shortlist model}.}
In~\cite{us}, a relaxation of the secretary problem  is introduced
where the algorithm is allowed to  select a {\it shortlist} of items that is larger than the number of items that ultimately need to be selected. after seeing the entire input, the algorithm can choose from the bigger set of items in the shortlist. 
This model is closely related to the random order streaming model. A comprehensive comparison between these two models can be found in~\cite{us}.
The main result of~\cite{us} is an online algorithm  for \nameOfProblemSL\ that, for any constant $\epsilon>0$, achieves a competitive ratio of $1-\frac{1}{e} - \epsilon-O(\frac{1}{k})$ with shortlist of size $ O(k)$. They also provide an implementation of their algorithm in the streaming setting with the same approximation ratio and memory $O(k)$. 
\paragraph{\textbf{The \nameOfProblemMatroidSL}.} Motivated by the improvements achieved for the competitive ratio of  \nameOfProblem\ in the \textit{shortlist} model, we ask if similar improvements can be achieved by relaxing the \nameOfProblemMatroid~to have a shortlist. That is, instead of choosing an independent set of a matroid $\mathcal{M}$ with $rk(\mathcal{M})=k$, the algorithm is allowed to chose  $\eta(k)$ items as part of a shortlist, for some function $\eta$;
and at the end of all inputs, the algorithm chooses an independent subset of items from the $\eta(k)$ selected items. Then what is the best competetive ratio that we can achieve 
in this model for example  when 
 $\eta(k)=O(k)$? Is it possible to improve the best known competetive ratio for matroid secretary problem in this model? 

\subsection{Problem definition}
We now give a more formal definition. We are given matroid $\mathcal{M}=(\mathcal{N},\mathcal{I})$, with $rk(\mathcal{M})=k$.
Items from a set ${\mathcal{ U}} = \{a_1, a_2, \ldots, a_n\}$  (pool of items) arrive in a uniformly random order over $n$ sequential rounds. The set ${\mathcal U}$ is apriori fixed but unknown to the algorithm, and the total number of items $n$ is known to the algorithm. In each round, the algorithm irrevocably decides whether to add the arriving item to a { shortlist} $A$ or not. 
The algorithm's value at the end of $n$ rounds is given by 
$$\ALG = \Ex[\max_{S\subseteq A, S\in \mathcal{I}} f(S)]$$ 
where $f(\cdot)$ is a monotone submodular function. The algorithm has value oracle access to this function.
The optimal offline utility is given by
$$\OPT:=f(S^*), \text{ where } S^*=\arg \max_{S \subseteq [n], S\in \mathcal{I}} f(S).$$ 
We say that an algorithm for this problem achieves a competitive ratio  $c$ using shortlist of size $\eta(k)$, if at the end of $n$ rounds, $|A|\le \eta(k)$ and $\frac{\ALG}{\OPT}\ge c$.

Given the shortlist $A$, since the problem of computing the solution $\arg \max_{S\subseteq A, S\in \mathcal{I}} f(S)$ can itself be computationally intensive, our algorithm will also track and output a subset $A^* \subseteq A, |A^*| \le k$. 

\subsection{Our Results}
We design an algorithm that achieves a $\frac{1}{2}(1-1/e^2-\epsilon-O(1/k))$ competitive ratio for any constant $\epsilon>0$, using a shortlist of size $O(k)$ for the matroid secretary problem with shortlists. This is especially surprising considering that the best known competitive ratio for the matroid secretary problem is $O(\log \log k)$. We are also able to get a constant competitive algorithm using shortlist of size  at most $k$ and also a constant competitive algorithm in the preemption model.

\begin{restatable*}{theorem}{matroidThm} \label{opttheorem}
For any constant $\epsilon>0$, there exists an online algorithm (Algorithm \ref{alg:main}) for the \nameOfProblemMatroidSL\ that achieves a competitive ratio of $\frac{1}{2}(1-\frac{1}{e^2} -\epsilon -O(\frac{1}{k}))$, with shortlist of size $\eta_\epsilon(k)=O(k)$. Here,  $\eta_\epsilon(k)=O(2^{poly(1/\epsilon)}k)$. The running time of this online algorithm is $O(nk)$.
\end{restatable*}

\begin{restatable*}{theorem}{thmpreemption} \label{thm:preemption}
For the matroid secretary problem in the preemption model, and matroid secretary problem 
that uses shortlist of size at most $\eta(k)=k$,
there is an algorithm 
that achieves a constant competitive ratio. 
\end{restatable*}

Furthermore,  for a more general constraint, namely $p$-matchoid constraints (defined in section~\ref{sec:matchoid}) we prove:

\begin{restatable*}{theorem}{matchoidThm} \label{opttheoremmatchoid}
For any constant $\epsilon>0$, there exists an online algorithm for the submodular secretary problem with $p$-matchoid constraints that achieves a competitive ratio of $\frac{1}{p+1}(1-\frac{1}{e^{p+1}} -\epsilon -O(\frac{1}{k}))$, with shortlist of size $\eta_\epsilon(k)=O(k)$. Here,  $\eta_\epsilon(k)=O(2^{poly(1/\epsilon)}k)$. The running time of this online algorithm is $O(n\kappa^{p})$, where $\kappa= \max_{i\in[q]} rk(\mathcal{M}_i)$.
\end{restatable*}

The  proposed  algorithm  also  has  implications  for  another  important  problem
of submodular function maximization under random order streaming model and matchoid constraints.
$\frac{1}{p+1}(1-1/e^{p+1}-\epsilon-O(1/k))$
approximation. 

\begin{restatable*}{theorem}{thmStreamingMatroid}
\label{thm:streamingMatroid}
For any constant $\epsilon\in (0,1)$, there exists an algorithm for the \streamingProblem\ with matroid constraints  that achieves $\frac{1}{2}( 1-\frac{1}{e} -\epsilon -O(\frac{1}{k}))$ approximation to $\OPT$ while using a memory buffer of size at most $\eta_\epsilon(k)=O(k)$. Also, the number of objective  function evaluations for each item, amortized over $n$ items, is $O(pk+\frac{k^2}{n})$.
\end{restatable*}

\begin{restatable*}{theorem}{thmStreamingMatchoid}
\label{thm:streamingMatcoid}
For any constant $\epsilon>0$, there exists an algorithm for the \streamingProblem\ with $p$-matchoid constraints  that achieves $\frac{1}{p+1}(1-\frac{1}{e^{p+1}} -\epsilon -O(\frac{1}{k}))$ approximation to $\OPT$ while using a memory buffer of size at most $\eta_\epsilon(k)=O(k)$. Also, the number of objective  function evaluations for each item, amortized over $n$ items, is $O(p\kappa+\kappa^p+\frac{k^2}{n})$, where $\kappa= \max_{i\in[q]} rk(\mathcal{M}_i)$.
\end{restatable*}

\subsection{ Related Work}
In this section, we overview some of the related online problems.
In the \textit{\nameOfProblem} introduced by \citet{Bateni} and ~\citet{Gupta:2010},
the algorithm selects $k$ items, but 
the value of the selected items is given by a monotone submodular function 
The algorithm can select at most $k$ items $S=\{a_1 \cdots, a_k\}$, from a randomly ordered sequence of $n$ items. The goal is to maximize
$f(S)$.
Currently, the best result for this setting is due to ~\citet{kesselheim}, who achieve a $1/e$-competitive ratio in exponential time in $k$, or $\frac{1}{e}(1-\frac{1}{e})$ in polynomial time in $n$ and $k$.
Submodular functions also has been used in the network design problems~\cite{ghugequasi,friggstad2016logarithmic}. 
There are also some related online coloring problems in the literature~\cite{gijswijt2007clique,abam2014online}.

\toRemove{
However, it is unclear if a competitive ratio of $1-1/e$ can be achieved by an online algorithm for the {\nameOfProblem} even when $k$ is large. 
}

In the \textit{matroid secretary problem}, the elements of a matroid $\mathcal{M}$ arrive in random order. Once we observe an item we need to irrevocably decide 
whether or not to accept it. The set of selected elements should form an independent set of the matroid.
The goal is to maximize the total sum of the values assigned to these elements.
It has applications in online mechanism design, in particular when the set of acceptable agents form a matroid~\cite{Babaioff:2008}. 

The existence of a constant competitive algorithm is a long-standing open problem.
~\citet{Lachish14} provides  the first $O(\log \log(k))$- competitive algorithm (the hidden constant is $2^{2^{34}}$).
~\citet{feldman2014simple} give a simpler  order-oblivious 
$2560 (\log\log 4k+5)$-competitive algorithm for the matroid secretary problem, by knowing only  the cardinality of the matroid in advance.
There are some $O(1)$-competitive algorithms for special variants
of the matroid secretary problem. 
For example, the elements of the ground set are assigned to a set of  weights uniformly at random hen a $5.7187$-competitive algorithm is possible for any matroid~\cite{soto2013matroid}.
Furthermore, a $16(1-1/e)$-competitive algorithm can be achieved as long as the weight assignment is  done at random,  even if we assume the adversarial arrival order.

 Recently,~\citet{Buchbinder:2014} considered a
different relaxation which is called preemptions model. In this model, elements added
to $S$ can be discarded later.
The main result of~\cite{Buchbinder:2014},  is a randomized $0.0893$-competitive
algorithm for cardinality constraints using $O(k)$ memory.



\section{Algorithm description}
\label{sec:alg}
Before describing our main algorithm we design a subroutine for 
a problem that we call it \textit{secretary problem with replacement}:
we are given a matroid $\mathcal{M}=(\mathcal{E},\mathcal{I})$ and an independent set $S\in \mathcal{I}$.
A pool of items
$I=(a_1,\cdots, a_N)$ arriving sequentially in a uniformly random order,
find an element $e$ from $I$ that can be added to $S$ after removing 
possibly one element $e'$ from $S$ such that the set remains independent, i.e., $S+e-e' \in \mathcal{I}$.
The goal is to choose element $e$ and $e'$ in an online manner with maximum marginal increment $g(e,S)=f(S+e-e')-f(S)$.
More precisely define function $g$ as:
\begin{equation}
\label{eq:g}
g(e,S):= f(S+e-\theta(e,S)) - f(S),
\end{equation}
where $\theta$ is defined as: 
$$ \theta(e,S) := \arg\max_{e'\in S} \{ f(S+e-e')| S+e-e' \in \mathcal{I} \} $$
We will consider the variant in which we are allowed to have a shortlist, where the  algorithm can add items to 
a shortlist and choose one item from the shortlist at the end.

For the \textit{secretary problem with replacement}, we give Algorithm~\ref{alg:matroidmax} which is a simple modification of the \textit{online max algorithm} in~\cite{us}.

\begin{lemma}
Algorithm~\ref{alg:matroidmax}, returns  element  $e$ with maximum  $g(e,S)$ with probability $1-\delta$, thus it achieves a  $1-\delta$ competitive ratio for the \textit{secretary problem with replacement}. 
The size of the shortlist that it uses is logarithmic in $1/\delta$.
\end{lemma}

\begin{algorithm*}[ht]
  \caption{~\bf{Secretary Problem with Replacement}}
  \label{alg:matroidmax} 
\begin{algorithmic}[1]
\State Inputs: number of items $N$, an independent set $S$, items in $I=\{a_1, \ldots, a_N\}$ arriving sequentially, $\delta \in (0,1]$. 
\State Initialize: $A\leftarrow \emptyset$,  $u=n\delta/2$, $M = -\infty$ 
\State $L \leftarrow 4\ln(2/\delta)$

\For {$i= 1$ to $N$}
\If {$ g(a_i,S) > M$}
\State $M \leftarrow  g(a_i,S)$
\If {$i \geq u$ and $|A|<L$}
\State $A\leftarrow A\cup \{a_i\}$
\EndIf
\EndIf
\EndFor
\State return $A$, and $A^*:= \max_{i\in A} g(a_i,S)$
\end{algorithmic}
\end{algorithm*}

\begin{algorithm*}[h!]
  \caption{~\bf{Algorithm for {\bf submodular} matroid secretary with shortlist}}
  \label{alg:main} 
  \label{alg:tmp} 
\begin{algorithmic}[1]
\State Inputs: number of items $n$, submodular function $f$, parameter $\epsilon \in (0,1]$. 
\State Initialize: $S_0 \leftarrow \emptyset, R_0 \leftarrow \emptyset, A \leftarrow \emptyset, A^* \leftarrow \emptyset$, constants $\alpha \ge 1, \beta \ge 1$ which depend on the constant $\epsilon$.
\State Divide indices $\{1,\ldots, n\}$ into $(\alpha, \beta)$ windows. 
\For {window $w= 1, \ldots, k/\alpha$} 

 \For {every slot $s_j$ in window $w$, $j=1,\ldots, \alpha\beta$}
  \State Concurrently for all subsequences of previous slots $\tau\subseteq \{s_1, \ldots, s_{j-1}\}$ of length $|\tau|<\alpha$ \label{li:subb}\\
  \hspace{0.44in} in window $w$, call the online algorithm in Algorithm \ref{alg:matroidmax} with the following inputs: 
  \begin{itemize}
  \item   number of items $N=|s_j|+1$, $\delta=\frac{\epsilon}{2}$, and
  \item item values $I=(a_0, a_1, \ldots, a_{N-1})$, with 
   
     \begin{eqnarray*} 
  a_0 & := & \max_{x\in R_{1,\ldots, w-1}} \Delta(x|S_{1,\ldots,w-1} \cup \gamma(\tau)\setminus \zeta(\tau)) \\
     a_\ell & := & \Delta(s_j(\ell)| S_{1,\ldots,w-1} \cup \gamma(\tau)\setminus \zeta(\tau) ),  \forall 0<\ell\le N-1
     \end{eqnarray*}
 where $s_j(\ell)$ denotes the $\ell^{th}$ item in the slot $s_j$. 
  \end{itemize}
\State Let $A_{j}(\tau)$ be the shortlist returned by  Algorithm \ref{alg:matroidmax} for slot $j$ and subsequence $\tau$. Add \\
\hspace{0.44in} all items except the dummy item $0$ to the shortlist $A$. 
 That is, \label{li:sube}
 $$A\leftarrow A\cup  (A(j)\cap s_j)$$
 \EndFor
 \State After seeing all items in window $w$, compute $R_w, S_w$ and $\bar{S}_w$ as before 
 \State $S_{1,\cdots, w} \leftarrow S_{1,\cdots, w-1}\cup S_w \setminus \bar{S}_w$
 \State $A^* \leftarrow A^*\cup (S_w \cap A)\setminus \hat{S}_w$
\EndFor
\State return $A$, $A^*$. 
\end{algorithmic}
\end{algorithm*}

Similar to~\cite{us}, we divide the input into sequential blocks that we refer to as $(\alpha, \beta)$ windows.
Here $k=rk(\mathcal{M})$.
\begin{definition}[$(\alpha,\beta)$ windows] \label{def:windows}
Let $X_1,\ldots,X_{k\beta}$ be a $(n,k\beta)$-ball-bin random set.
Divide the indices $\{1,\ldots, n\}$ into $k\beta$ slots, where the $j$-th slot, $s_j$, consists of $X_j$ consecutive indices in the natural way, that is, slot $1$ contains the first $X_1$ indices, slot $2$ contains the next $X_2$, etc.
 Next, we define $k/\alpha$ windows, where window $i$ consists of $\alpha \beta$ consecutive slots, in the same manner as we assigned slots.
\end{definition}



Intuitively, for large enough $\alpha$ and $\beta$,
roughly $\alpha$ items from the optimal set $S^*$ are likely to lie in each of these windows, and further, it is unlikely that two items from $S^*$ will appear in the same slot. 

The algorithm can focus on identifying a constant number (roughly $\alpha$) of optimal items from each of these windows, with at most one item coming from each of the $\alpha \beta$ slots in a window. 
Similar to~\cite{us}, the core of our algorithm is a subroutine that accomplishes this task in an online manner using a shortlist of constant size in each window. But the difference is that adding items from a new window to the current solution $S$ could make it a non-independent set of $\mathcal{M}$. In order to make the new set independent we have to remove some items from $S$. The removed item corresponding to $e$ will be $\theta(e,S)$. We need to take care of all the removals for newly selected items in a window. Therefore we have to slightly change the definitions in~\cite{us}. We introduce $\zeta(\tau)$ which is counterpart of $\gamma(\tau)$ for the removed elements. 
More precisely, for any subsequence $\tau=(s_1,\ldots, s_\ell)$ of the $\alpha\beta$ slots in window $w$, recall the greedy subsequence $\gamma(\tau)$ of items as:
\begin{equation}
\label{eq:gamma}
\gamma(\tau):=\{i_1, \ldots, i_\ell\}
\end{equation}
where
\begin{equation}
\label{eq:ij}
i_j := \arg \max_{i\in s_j \cup R_{1, \ldots, w-1}} g(i, S_{1, \ldots,w-1} \cup \{i_1,\ldots, i_{j-1}\})
\end{equation}
now define $\zeta(\tau):=\{c_1, \ldots, c_\ell\}$
where
\begin{equation}
\label{eq:cij}
c_j :=  \theta(i_j, S_{1, \ldots,w-1} \cup \{i_1,\ldots, i_{j-1}\})
\end{equation}


Recall the definition of $R_w$ in~\cite{us}, which is  the union of all greedy subsequences of length $\alpha$, and $S_w$ to be  the best subsequence among those. That is,
\begin{equation}
\label{eq:Rw}
R_w = \cup_{\tau: |\tau|=\alpha} \gamma(\tau)
\end{equation}
and 
\begin{equation}
\label{eq:Sw}
S_w=\gamma(\tau^*),
\end{equation} now define 
\begin{equation}
\label{eq:Swbar}
\bar{S}_w=\zeta(\tau^*),
\end{equation}
where
\begin{equation}
\label{eq:taustar}
\tau^*:=\arg \max_{\tau: |\tau|=\alpha} f((S_{1,\ldots,w-1} \cup \gamma(\tau))\setminus \zeta(\tau)) - f(S_{1,\ldots,w-1}).
\end{equation}
also define
\begin{equation}
\label{eq:Swhat}
\hat{S}_w=\{c_{j_1},\cdots, c_{j_t}\}, \text{where } (S_w\cap A)=  \{i_{j_1},\cdots, i_{j_t}\}
\end{equation}

In other words, $\hat{S}_w$ is counterpart of elements of $S_w\cap A$ that are removed by $g$. 
Also note that in the main Algorithm~\ref{alg:main}, we  remove $\zeta(\tau^*)$ from $S_{1, \cdots, w-1}\cup S_w$ at the end of window $w$ and make $S_{1,\cdots, w}$ an independent set of $\mathcal{M}$.


In order to find the item with the maximum $g$ value \eqref{eq:ij}, among all the items in the slot. 
We use an online subroutine that employs the algorithm (Algorithm \ref{alg:matroidmax}) for the \textit{secretary problem with replacement} described earlier. 
Note that $R_w$, $S_w$ and $\bar{S}_w$ can be computed {\it exactly at the end} of window $w$.



The algorithm returns both the shortlist $A$ which similar to~\cite{us} is of size $O(k)$ as stated in the following proposition, as well the set $A^*$.
Note that we remove $\hat{S}_w$ from $A^*$ at the end of window $w$.
In the next section, we will show that $\Ex[f(A^*)] \ge (1-\frac{1}{e^2} - \epsilon -O(\frac{1}{k})) f(S^*)$ to provide a bound on the competitive ratio of this algorithm.  As it is proved in~\cite{us},

\begin{prop}
\label{prop:size}
Given $k,n$, and any constant $\alpha, \beta$ and $\epsilon$, the size of shortlist $A$ selected by Algorithm~\ref{alg:main} is at most $4k \beta {\alpha \beta \choose \alpha}\log(2/\epsilon) = O(k)$. 
\end{prop}

\toRemove{
\section{Bounding the size of shortlist}
\label{sec:size}
Given the algorithm description, it is not very difficult to show that 
our algorithm uses a shortlist of size at most $O(k)$ if $\alpha, \beta$ are constants. 
The number of items selected by Algorithm \ref{alg:SIIImax} is at most $L$.
By choosing appropriate size $L$ for shortlist $A$,
w.h.p. we can gaurantee  that $A$ contains the maximum element $M$.
\scomment{Change $\epsilon,\delta$ to get max with probability $1-\epsilon$.}


Now, we can deduce the desired bound on the size of shortlist returned by Algorithm \ref{alg:tmp}
\begin{prop}
The size of shortlist $A$ selected by Algorithm~\ref{alg:tmp} is $O(k)$.
\end{prop}
\begin{proof}
Note that for each window $w=1,\ldots, k/\alpha$, and for each of the $\alpha\beta$ slots in this window, the subroutine in Algorithm \ref{alg:tmpWindow} runs  Algorithm \ref{alg:SIIImax} for ${\alpha \beta \choose \alpha}$ times (for all $\alpha$ length subsequences) to add at most $\log(1/\delta)$ items each time to . Therefore, the above algorithm selects $ k \beta {\alpha \beta \choose \alpha} = O(k)$ items as part of the shortlist.
\end{proof}

}
\toRemove{
\subsection{Bounding the size of shortlist for Algorithm \ref{alg:SIIImax}}
\scomment{Can we add a lemma showing that expected number of items selected will be $\log(1/\delta)$ when $u=N\delta$}
\mcomment{same as previous comment}
\scomment{Also we need to discuss that when we say shortlist of size $\eta(k)$, is that expected number of items or high probability? If it is high probability, we need to define what it means to say shortlist of size $\eta(k)$. Does it mean $\eta(k)$ with prescribed probability $1-\epsilon'$?}
\scomment{the results here will move to the previous section, and proofs for high prob bound will go to appendix}

\begin{theorem}
$\mathbb{E}[|S|] = \ln n$.
\end{theorem}
\begin{proof}
Suppose $S_i=\{a_1, \cdots, a_i\}$.
Consider all permutations of $S_i$, an element will be selected at position $i$ if it is equal to $\max_{a\in S_i} F(a)$, thus the probability that it gets selected is $1/i$. Therefore the expected number of selections will be at most $\sum_{i=1}^{n} \frac{1}{i}  = \ln n $
\end{proof}


\scomment{What is $\delta$ in the theorem below? The proof below seems to suggest it is high probability statement, not exact.}
\begin{theorem}\label{maxanalysis}
Algorithm~\ref{alg:SIIImax}, with parameter $u=n\epsilon$,  selects a set $S$ with 
$$|S|<\ln (1/\epsilon)+\ln(1/\delta)+\sqrt{\ln^2{1/\delta}+2m\ln(1/\delta)\ln(1/\epsilon)}$$ 
and $\mathbb{E}[\max_{a\in S} a]=(1-\epsilon-\delta)OPT$, where $OPT$ is the max element in the input.
\end{theorem}
\begin{proof}
We use Freedman's inequality.
 If $\{a_1,\cdots, a_i\}$ has a unique maximum, 
define $Y_i$ to be a random variable indicating whether the algorithm has selected $a_i$  or not, where $Y_i=1-\frac{1}{i}$ if $a_i$ is selected and $Y_i=-\frac{1}{i}$ otherwise. 
 If it has not unique solution define $Y_i=0$. ($a_i$ will not be selected)
 Also define $\mathcal{F}_i=\{Y_{n},Y_{n-1}, \cdots, Y_{n-i+1}\}$.

Let $X_i=\sum_{j=n-i+1}^{n} Y_j$,
then $\{X_i\}$ is a martingle, 
because $E[X_{i+1}|\mathcal{F}_{i}] = X_i+E[Y_{n-i}|\mathcal{F}_i]$.
 If $\{a_1,\cdots, a_i\}$ has a unique maximum element, $E[Y_{n-i}|\mathcal{F}_i]=(1/i)(1-1/i)+(1-1/i)(-1/i)=0$,
 otherwise $E[Y_{n-i}|\mathcal{F}_i]=0$. So in both cases $E[X_{i+1}|\mathcal{F}_{i}] =X_i$.
As in the Freedman's inequality, let $L=\sum_{i=n\epsilon}^{n} Var(Y_i| F_{i-1})$. 
\begin{align*}
L  =  \sum_{i=n\epsilon}^{n} \frac{1}{i}  (1-\frac{1}{i})^2 + (1-\frac{1}{i}) (\frac{1}{i})^2 
< \sum_{i=n\epsilon}^{n} \frac{1}{i} =\ln (1/\epsilon)
\end{align*}
Therefore,
\[
Pr(X_{n-n\epsilon}\ge \alpha \text{ and }  L\le \ln (1/\epsilon) ) \le exp(-\frac{\alpha^2}{\ln (1/\epsilon)+ 2\alpha })   < \delta 
\]
Thus we get $\alpha > \ln(1/\delta)+\sqrt{\ln^2{1/\delta}+2\ln(1/\delta)\ln(1/\epsilon)}$.
Also $|S| = X_{n-n\epsilon} + \ln(1/\epsilon)$. Therefore 
\[
Pr(|S| \ge  \ln (1/\epsilon)+\ln(1/\delta)+\sqrt{\ln^2{1/\delta}+2\ln(1/\delta)\ln(1/\epsilon)}  )  \le \delta
\]
So with probability $(1-\delta)$, $|S| \le \ln (1/\epsilon)+\ln(1/\delta)+\sqrt{\ln^2{1/\delta}+2\ln(1/\delta)\ln(1/\epsilon)} $.  
Also $\mathbb{P}(OPT \in\{a_{n\epsilon},\cdots, a_n \}  ) =(1-\epsilon)$.
Therefore $E[\max_{a\in S} a] \ge (1-\epsilon)OPT-\delta OPT$. 
\end{proof}

\begin{theorem}
Any online algorithm needs to select at least $\frac{1}{2}\log(1/\epsilon)-\frac{1}{2}$ elements, in expectation, to select the maximum element  with probability at least $(1-\epsilon)$ in a random permutation. (we assume $n> 1/\epsilon$)
\end{theorem}
\begin{proof}
Let $I_i=\{a_1,\cdots, a_{n/2^{i-1}}\}$, $T_i=\{a_{n/2^i+1}, \cdots, a_{n/2^{i-1}}\}$, and $R_i=I_1\setminus I_i$, for $i=1, \cdots, \log(1/\epsilon)$. Suppose $M_i$ is the maximum element in $I_i$. Let $S$  be the set of selected elements by algorithm at the end of execution.
Suppose $\epsilon_i= E[M_i\notin S| M_i\in T_i]$,
then $E[|S\cap T_i|] \ge \frac{1}{2} (1-\epsilon_i)$.
Therefore $E[|S|] \ge \sum_{i=1}^{\log(1/\epsilon)} \frac{1}{2} (1-\epsilon_i) $. Also w.p. $\frac{1}{2^{i}}$, $M_1 \in T_i$, 
thus $\sum_{i=1}^{\log(1/\epsilon)} \frac{1}{2^{i}} \epsilon_i  \le \epsilon$. 
(Note that we use the fact $ E[M_i\notin S| M_i\in T_i \text{ and }  M_i=M_1] \le \epsilon_i$, i.e, if algorithm selects one element it will select 
it even if we increase its value and keep the rest untouched) 
Now $E[|S|]$ is minimized under above constraint if $\frac{1}{2^{\log(1/\epsilon)}}\epsilon_{\log(1/\epsilon)} = \epsilon$ and the rest are zero.
Hence $E[|S|] \ge \frac{1}{2}\log(1/\epsilon)-\frac{1}{2}$.
\end{proof}


}

\section{Preliminaries}
\label{sec:analysis}

 The following properties of submodular functions are well known (e.g., see~\cite{Buchbinder:2014,Feige:2011,Feldman:2015}). 
\begin{lemma}\label{marginalsum}
Given a monotone submodular function $f$, and subsets $A,B$ in the domain of $f$, we use $\Delta_f(A|B)$ to denote $f(A\cup B)-f(B)$. 
For any set $A$ and $B$, $\Delta_f(A|B) \le \sum_{a\in A\setminus B} \Delta_f(a|B)$
\end{lemma}
\begin{lemma}\label{sample}
Denote by $A(p)$ a random subset of $A$ where each element has probability at least $p$ to appear in $A$ (not necessarily independently). Then $E[f(A(p))] \ge (1-p) f(\emptyset) + (p)f(A)$
\end{lemma}


We will use the following well known deviation inequality for martingales (or supermartingales/submartingales).

\begin{lemma}[Azuma-Hoeffding inequality]
\label{lem:azuma}
Suppose $\{ X_k : k = 0, 1, 2, 3, ... \}$ is a martingale (or super-martingale) and ${\displaystyle |X_{k}-X_{k-1}|<c_{k},\,} $
almost surely. Then for all positive integers N and all positive reals $r$,
$${\displaystyle P(X_{N}-X_{0}\geq r)\leq \exp \left(\frac{-r^{2}}{2\sum _{k=1}^{N}c_{k}^{2}}\right).} $$
And symmetrically (when $X_k$ is a sub-martingale):
$$
{\displaystyle P(X_{N}-X_{0}\leq -r)\leq \exp \left(\frac{-r^{2}}{2\sum _{k=1}^{N}c_{k}^{2}}\right).} 
$$
\end{lemma}
\begin{lemma}[Chernoff bound for Bernoulli r.v.]
\label{lem:Chernoff}
Let $X = \sum_{i=1}^N X_i$, where $X_i = 1$ with probability $p_i$ and $X_i = 0$ with
probability $1 - p_i$, and all $X_i$ are independent. Let $\mu = \Ex(X) = \sum_{i=1}^N p_i$. Then, 
$$P(X \geq (1 + \delta)\mu) \le e^{{-\delta^2\mu}/{(2+\delta)}} $$
for all $\delta>0$, and
$$P(X \leq (1 - \delta)\mu) \le e^{-\delta^2\mu/2}$$
for all $\delta\in(0,1)$. 
\end{lemma}

\begin{definition}
(\textbf{Matroids}). A matroid is a finite set system $\mathcal{M} = (\mathcal{N} , \mathcal{I})$, where $\mathcal{N}$ is a set and $\mathcal{I}\subseteq 2^{\mathcal{N}}$
is a family of subsets such that: (i) $\emptyset \in I$, (ii) If $A \subseteq B \subseteq N$ , and $B \in I$, then $A \in I$,
(iii) If $A, B \in I$ and $|A| < |B|$, then there is an element $b \in B\setminus A$ such that $A + b \in I$. In a matroid $\mathcal{M} = (\mathcal{N} , \mathcal{I})$, $N$ is called the ground set and the members of $\mathcal{I}$ are called independent sets of the matroid. The bases of $\mathcal{M}$ share a common cardinality, called the
rank of $\mathcal{M}$.
\end{definition}

\begin{lemma} 
\label{lem:Brualdi} (\textbf{Brualdi} )
If $A, B$ are any two bases of matroid $M$ then there exists a
bijection $\pi$ from $A$ to $B$, fixing $A\cap B$, such that $A - x + \pi(x) \in M$ for all $x \in A$.
\end{lemma}

\subsection{Some useful properties of $(\alpha, \beta)$ windows}
In~\cite{us}, we proved some useful properties of $(\alpha, \beta)$ windows, defined in Definition~\ref{def:windows} and used in Algorithm \ref{alg:main}. 
which we summarize it in this section.

The first observation is that every item will appear uniformly at random in one of the $k\beta$ slots in $(\alpha,\beta)$ windows. 

\begin{definition}
For each item $e\in I$, define $Y_e\in [k\beta]$ as the random variable indicating the slot in which $e$ appears. We call vector $Y\in [k\beta]^n$ a \textit{configuration}.
\end{definition}

\begin{lemma} \label{lem:indep}
Random variables $\{Y_e\}_{e\in I}$ are i.i.d. with uniform distribution on all $k\beta$ slots. 
\end{lemma}

This follows from the uniform random order of arrivals, and the use of the balls in bins process to determine the number of items in a slot during the construction of $(\alpha,\beta)$ windows. 

Next, we make important observations about the probability of assignment of items in $S^*$ in the slots in a window $w$, given the sets $R_{1,\ldots, w-1}, S_{1,\ldots, w-1}$ (refer to \eqref{eq:Rw}, \eqref{eq:Sw} for definition of these sets). 
To aid analysis, we define the following new random variable $T_w$ that will track all the useful information from a window $w$.  
\begin{definition}
Define $T_w:= \{(\tau, \gamma(\tau))\}_\tau$, for  all $\alpha$-length subsequences $\tau=(s_1,\ldots, s_\alpha)$ of the $\alpha\beta$ slots in window $w$. Here, $\gamma(\tau)$ is a sequence of items as defined in \eqref{eq:gamma}. Also define  $Supp(T_{1,\cdots ,w}) :=\{e| e\in\gamma(\tau) \text{ for some } (\tau, \gamma(\tau))\in T_{1,\cdots, w} \} $ (Note that $Supp(T_{1,\cdots, w})=R_{1,\ldots, w}$).
\end{definition}

\begin{lemma}
\label{config}
For any window $w\in [W]$,  $T_{1,\ldots, w}$ and $S_{1, \ldots, w}$ \toRemove{are uniquely defined for each configuration $Y$.}
are independent of the ordering of elements within any slot, and 
are determined by the configuration $Y$.
\end{lemma}
\toRemove{
\begin{proof}
Given the assignment of items to each slot, it follows from the definition of $\gamma(\tau)$ and $S_w$ (refer to \eqref{eq:gamma} and \eqref{eq:Sw}) that $T_{1,\ldots, w}$ and $S_{1, \ldots, w}$ are independent of the ordering of items within a slot. Now, since  the assignment of items to slot are determined by the configuration $Y$, we obtain the desired lemma statement.
\end{proof}
}

Following the above lemma, given a configuration $Y$, we will some times use the notation $T_{1,\ldots, w}(Y)$ and $S_{1, \ldots, w}(Y)$ to make this mapping explicit.

\begin{lemma}
\label{lem:pijBound}
For any item $i\in S^*$, window $w \in \{1,\ldots, W\}$, and slot $s$ in window $w$, define
\begin{equation}
\label{eq:pijBound}
p_{is}:=\mathbb{P}(i \in s \cup Supp(T) | T_{1,\ldots, w-1}=T).
\end{equation}
Then, for any pair of slots $s',s''$ in windows $w, w+1, \ldots, W$,
\begin{equation}
p_{is'}=p_{is''} \ge \frac{1}{k\beta} \ .
\end{equation}
\end{lemma}

\begin{lemma}
\label{lem:ijindep}
For any window $w$,  $i,j\in S^*, i\ne j$ and $s,s'\in w$, the random variables $\mathbf{1}(Y_i=s|T_{1,\cdots, w-1}=T)$ and $ \mathbf{1}(Y_j=s'|T_{1,\cdots, w-1}=T)$ are independent. That is, given $T_{1,\cdots, w-1}=T$,  items $i,j\in S^*, i\ne j$ appear in any slot $s$ in $w$ independently.
\end{lemma}
\toRemove{
\begin{proof}
To prove this, we show that $\mathbb{P}(Y_i=s|T_{1,\cdots, w-1}=T)=\mathbb{P}(Y_i=s|T_{1,\cdots, w-1}=T \text{ and } Y_j =s')$. Suppose $Y'$ is a configuration such that $Y'_i=s$ and $Y'_j=s'$, and $T_{1,\cdots, w-1}(Y')=T$. Assume there exists another feasible slot assignment of $j$, i.e., there is another configuration $Y''$ such that $T_{1,\cdots, w-1} (Y'') =T$ and $Y''_j=s''$ where $s''\ne s'$. (If no such configuration $Y''$ exists, then $\mathbf{1}(Y_j=s')|T$ is always $1$, and the desired lemma statement is trivially true.) Then, we prove the desired independence by showing that there exists a feasible configuration where slot assignment of $i$ is $s$, and $j$ is $s''$. This is obtained by changing $Y_j$ from $s'$ to $s''$ in $Y'$, to obtain another configuration $\bar{Y}$. In Lemma~\ref{addtoslot}, we show that this change will not effect $T_{1,\cdots, w-1}$, i.e., $T_{1,\cdots, w-1} (\bar Y)=T $. 
Thus configuration $\bar Y$ satisfies the desired statement.
\end{proof}
}

\begin{lemma} \label{addtoslot}
Fix a slot $s'$, $T$, and $j\notin Supp(T)$. Suppose that there exists  some configuration  $Y'$ such that $T_{1,\cdots, w-1} (Y') =T$ and $Y_j'=s'$. Then, given any configuration $Y''$ with $T_{1,\ldots, w-1}(Y'')=T$, we can replace  $Y''_j$ with $s'$ to obtain a new configuration $\bar Y$ that also satisfies $T_{1,\ldots, w-1}(\bar Y)=T$.
\toRemove{
Suppose there exists at least one configuration $\bar{Y}$ such that   $T_{1,\cdots, w-1} (\bar{Y}) =T$ and $\bar{Y}_j=s$ for $j\notin Supp(T)$.
Then for any configuration $Y$ with  $T_{1,\cdots, w-1} (Y) =T$, by setting $Y_j=s$, we get a new configuration $Y'$
such that $T_{1,\cdots, w-1} (Y') =T$.}
\end{lemma}
\section{Analysis of the algorithms}
\label{sec:analysis}
\newcommand{\settinga}{\mbox{$k \ge \alpha\beta$}, \settingb}
\newcommand{\settingb}{\mbox{$\beta\ge \frac{8}{(\delta')^2}$}, $\alpha\ge 8\beta^2 \log(1/\delta')$}

In this section we show that for any $\epsilon \in (0,1)$, Algorithm \ref{alg:tmp} with an appropriate choice of constants $\alpha, \beta$, achieves the competitive ratio claimed in Theorem \ref{opttheorem} for the \nameOfProblemMatroidSL.  

First, we use the observations from the previous sections to show the existence of a random subsequence of slots $\tilde \tau_w$ of window $w$ such that we can lower bound $f((S_{1,\ldots, w-1}\cup \gamma(\tilde \tau_w))\setminus \zeta(\tilde \tau_w))- f(S_{1,\ldots, w-1})$ in terms of $\OPT -f(S_{1,\ldots, w-1})$. This will be used to lower bound  increment 
$f(S_{1,\ldots, w-1}\cup \gamma(\tau^*)\setminus \zeta(\tau^*)) -f(S_{1,\ldots, w-1})$ in every window.


\begin{definition}[$Z_s$ and $\tilde \gamma_w$]
\label{def:tauw}
Create sets of items $Z_s, \forall s\in w$  as follows: for every slot $s$, add every item from $i\in S^*\cap s$ independently with probability $\frac{1}{k \beta p_{is}}$ to $Z_s$. Then, for every item $i\in S^*\cap T$, with probability $\alpha/k$, add $i$ to $Z_s$ for a randomly chosen slot $s$ in $w$. Define subsequence $\tilde \tau_w$ as the sequence of 
slots with $Z_s\ne \emptyset$. 

\end{definition}
Similar to~\cite{us}, we have the following property for $Z_s$:
\begin{lemma}
\label{lem:Zs}
Given any $T_{1,\ldots, w-1}=T$, for any slot $s$ in window $w$, all $i, i' \in S^*, i\ne i'$ will appear in $Z_s$ independently with probability $\frac{1}{k\beta}$. Also, given $T$, for every $i \in S^*$, the probability to appear in $Z_s$ is equal for all slots $s$ in window $w$. Further, each $i\in S^*$ occurs in  $Z_s$ of at most one slot $s$.
\end{lemma}

\begin{lemma}
We can show that for all $i, i' \in S^*\backslash \{Z_{s_1} \cup \ldots \cup Z_{s_{j-1}}\}$,
\begin{equation} 
\label{eq:pk}
\Pr(i\in Z_{s_j}| Z_{s_1} \cup \ldots \cup Z_{s_{j-1}}) = \Pr(i'\in Z_{s_j}| Z_{s_1} \cup \ldots \cup Z_{s_{j-1}})  \ge \frac{1}{k} \ .
\end{equation}
\end{lemma}
\begin{proof}
The proof is similar to Lemma 12 in~\cite{us}, and it is based on Lemma \ref{lem:Zs},
\end{proof}

In the following lemma we lower bound the marginal gain of a randomly picked element of optimal solution in slot $s_j$ with respect to previously selected items.

\begin{lemma}
\label{lem:asGoodas}

Suppose the sequence $\tilde \tau_w=(s_1, \ldots, s_t)$ defined as in Definition \ref{def:tauw}, let $\gamma(\tilde \tau_s)=(i_1,\ldots, i_t)$, with $\gamma(\cdot)$ as defined in \eqref{eq:gamma}. Then, for all $j=1,\ldots, t$, 

{\small \begin{eqnarray*}
\Ex[\Delta_f(a|S_{1,\ldots, w-1} \cup \{i_1, \ldots, i_{j-1}\}\setminus\{c_1,\cdots, c_{j-1}\})|T_{1,\ldots, w-1}, i_{1,\ldots, j-1}, a\in S^*\cap Z_{s_j}]\\
 \ge \frac{1}{k}\left((1-\frac{\alpha}{k})f(S^*)-f(S_{1,\ldots, w-1} \cup \{i_1, \ldots, i_{j-1}\}\setminus\{c_1,\cdots,c_{j-1}\})\right)\
\end{eqnarray*}}
\end{lemma}
\begin{proof}
We can lower bound the increment assuming $a$ is randomly picked item from $Z_{s_j}\cap S^*$:
\begin{eqnarray*} 
& & \Ex[\Delta_f(a,S_{1,\ldots, w-1} \cup \{i_1, \ldots, i_{j-1}\}\setminus\{c_1,\cdots,c_{j-1}\})|T_{1,\ldots, w-1}=T, i_1, \ldots, i_{j-1},a\in S^*\cap Z_{s_j}] \\
& \ge & \frac{1}{k} \Ex[\sum_{a\in S^*\backslash \{Z_1,\ldots Z_{s_{j-1}}\}} \Ex[\Delta_f(a,S_{1,\ldots, w-1} \cup \{i_1, \ldots, i_{j-1}\}\setminus\{c_1,\cdots,c_{j-1}\})|T, i_1, \ldots, i_{j-1}]] \\
& \ge & \frac{1}{k} \Ex[ \left(f(S^*\backslash  \{Z_1,\ldots Z_{s_{j-1}}\})-f(S_{1,\ldots, w-1} \cup i_{1,\ldots, s-1}\setminus\{c_1,\cdots,c_{j-1}\})\right) | T]\\
& \ge & \frac{1}{k}  \Ex[\left(f(S^*\backslash  \cup_{s'\in w} Z_{s'})-f(S_{1,\ldots, w-1} \cup i_{1,\ldots, s-1}\setminus\{c_1,\cdots,c_{j-1}\})\right)| T]\\
& \geq & \frac{1}{k} \left(\left(1-\frac{\alpha}{k}\right)f(S^*) -f(S_{1,\ldots, w-1} \cup i_{1,\ldots, s-1}\setminus \{c_1,\cdots, c_{j-1}\})\right)
\end{eqnarray*}
The last inequality uses 
Lemma \ref{sample} for submodular function $f$. 
and the observation from Lemma \ref{lem:Zs} that given $T$, every $i\in S^*$ appears in $\cup_{s'\in w} Z_{s'}$ independently with probability $\alpha/k$, so that every $i\in S^*$ appears in $S^*\backslash\cup_{s'\in w} Z_{s'}$ independently with probability $1-\frac{\alpha}{k}$; 
\end{proof}

\begin{lemma}\label{replaced}
Suppose the sequence $\tilde \tau_w=(s_1, \ldots, s_t)$ defined as in Definition \ref{def:tauw}, let $\gamma(\tilde \tau_s)=(i_1,\ldots, i_t)$, with $\gamma(\cdot)$ as defined in \eqref{eq:gamma}. 
Moreover, let $S'$ be the extension of $S_{1,\cdots, w-1}\cup \{i_1, \cdots, i_{j-1}\}\setminus \{c_1,\cdots, c_{j-1}\}$ to an independent set in $\mathcal{M}$, and $\pi$ be the bijection from Brualdi lemma (refer to Lemma~\ref{lem:Brualdi}) from $S^*$ to $S'$.
Then, for all $j=1,\ldots, t$, \\
{\small \begin{eqnarray*}
\Ex[f(S_{1,\ldots, w-1} \cup \{i_1, \ldots, i_{j-1}\}\setminus\{c_1,\cdots, c_{j-1},\pi(a)\})|T_{1,\ldots, w-1}, i_{1,\ldots, j-1}, a\in S^*\cap Z_{s_j}]\\
 \ge (1-\frac{1}{k-\alpha})f(S_{1,\ldots, w-1} \cup \{i_1, \ldots, i_{j-1}\}\setminus\{c_1,\cdots,c_{j-1}\})\
\end{eqnarray*}}
\toRemove{
{\small \begin{eqnarray*}
\Ex[\Delta_f(a|S_{1,\ldots, w-1} \cup \{i_1, \ldots, i_{j-1}\}\setminus\{c_1,\cdots, c_{j-1},\pi(a)\})|T_{1,\ldots, w-1}, i_{1,\ldots, j-1}, a\in S^*\cap Z_{s_j}]\\
 \ge (1-\frac{1}{k})f(S^*)-f(S_{1,\ldots, w-1} \cup \{i_1, \ldots, i_{j-1}\}\setminus\{c_1,\cdots,c_{j-1}\}\
\end{eqnarray*}}
}
\end{lemma}
\begin{proof}
Since $\pi$ is a bijection from $S^*$ to $S'$, from Brualdi's lemma (lemma~\ref{lem:Brualdi}), there is an onto mapping $\pi'$ from $S^*$ 
 to $S_{1,\cdots, w-1}\cup \{i_1, \cdots, i_{j-1}\}\setminus \{c_1,\cdots, c_{j-1}\}\cup \{\emptyset\}$ 
such that 
$S_{1,\cdots, w-1}\cup \{i_1, \cdots, i_{j-1}\}\setminus \{c_1,\cdots, c_{j-1}\}- \pi'(a) +a \in M$, for all $a \in S^*$. Further, $\pi'(a)=\pi(a)$ if $\pi(a)\in  S_{1,\cdots, w-1}\cup \{i_1, \cdots, i_{j-1}\}\setminus \{c_1,\cdots, c_{j-1}\}$ and $\pi'(a)=\emptyset$ otherwise.


Recall the definition of $Z_{s_j}$.
Suppose $a$ is a randomly picked item from $S^*\cap Z_{s_j}$.
Note that from Lemma~\ref{lem:Zs}, conditioned on $T_{1,\cdots, w-1}$, the element $a$ can be equally any element of 
$S^*\backslash  \{Z_1,\ldots Z_{s_{j-1}}\}$
with probability at least $1/(k-\alpha)$.
Therefore, $\pi'(a)$ would be any of $S_{1,\cdots, w-1}\cup \{i_1, \cdots, i_{j-1}\}\setminus \{c_1,\cdots, c_{j-1}\}$ with probability at most $1/(k-\alpha)$ 
(since $\pi'$ might map some elements of $S^*$ to the empty set).
Now based on the definition of $\pi$ and lemma~\ref{sample} we have:

{\small \begin{eqnarray*}
\Ex_a[f(S_{1,\ldots, w-1} \cup \{i_1, \ldots, i_{j-1}\}\setminus\{c_1,\cdots, c_{j-1},\pi(a)\})|T_{1,\ldots, w-1}, i_{1,\ldots, j-1}, a\in S^*\cap Z_{s_j}]\\
 \ge (1-\frac{1}{k-\alpha})f(S_{1,\ldots, w-1} \cup \{i_1, \ldots, i_{j-1}\}\setminus\{c_1,\cdots,c_{j-1}\})\
\end{eqnarray*}}

\toRemove{
{\small \begin{eqnarray*}
\Ex_a[\Delta_f(a|S_{1,\ldots, w-1} \cup \{i_1, \ldots, i_{j-1}\}\setminus\{c_1,\cdots, c_{j-1},\pi(a)\})|T_{1,\ldots, w-1}, i_{1,\ldots, j-1}, a\in S^*\cap Z_{s_j}]\\
 \ge (1-\frac{1}{k})f(S^*)-f(S_{1,\ldots, w-1} \cup \{i_1, \ldots, i_{j-1}\}\setminus\{c_1,\cdots,c_{j-1}\}\
\end{eqnarray*}}
}

\end{proof}

\begin{lemma}\label{exg}
Suppose function $g$ is as defined in equation~\ref{eq:g}.
For the sequence $\tilde \tau_w=(s_1, \ldots, s_t)$, and $\gamma(\tilde \tau_s)=(i_1,\ldots, i_t)$. Then, for all $j=1,\ldots, t$, \\
{\small \begin{eqnarray*}
\Ex[g(i_j,S_{1,\ldots, w-1} \cup \{i_1, \ldots, i_{j-1}\}\setminus\{c_1,\cdots, c_{j-1}\})|T_{1,\ldots, w-1}, i_{1,\ldots, j-1}]\\
 \ge 
 \frac{1}{k}
 \left(
 (1-\frac{\alpha}{k-\alpha})f(S^*)-
 2f(S_{1,\ldots, w-1} \cup \{i_1, \ldots, i_{j-1}\}\setminus\{c_1,\cdots,c_{j-1}\})\right)\
\end{eqnarray*}}

\toRemove{
$$ 
\Ex[g(i_j,{S}_{1,\ldots, w-1} \cup \{i_1,\cdots,i_{j-1}\}\setminus \{c_1,\cdots, c_{j-1}\})|T_{1,\cdots, w-1},i_{1,\cdots, j-1}]
$$
$$
\ge
{ \frac{1}{k}\left((1-\frac{\alpha}{k})f(S^*)-2f({S}_{1,\ldots, w-1} \cup \{i_1,\cdots,i_{j-1}\}\setminus \{c_1,\cdots, c_{j-1}\})\right)\ . }
$$  
}
\end{lemma}

\begin{proof}
In the algorithm~\ref{alg:main}, at the end of window $w$, we set 
$S_{1,\cdots, w} = S_{1,\cdots, w-1}\cup S_w \setminus \bar{S}_w$.
Suppose $a\in s_j\cap S^*$. 
Moreover, let $S'$ be the extension of $S_{1,\cdots, w-1}\cup \{i_1, \cdots, i_{j-1}\}\setminus \{c_1,\cdots, c_{j-1}\}$ to an independent set in $\mathcal{M}$, and $\pi$ be the bijection from Brualdi lemma (refer to Lemma~\ref{lem:Brualdi}) from $S^*$ to $S'$.
Thus the expected value of the function $g$ on the element selected by the algorithm in slot $s_j$ (the element with maximum $g$ in the slot $s_j$) would be
\begin{eqnarray*}
&\Ex[f({S}_{1,\ldots, w-1} \cup \{i_1, \ldots, i_{j}\} \setminus \{c_1,\cdots, c_j \}|T_{1,\ldots, w-1}, i_{1,\ldots, j-1}]  
\\\ge&
\Ex[f({S}_{1,\ldots, w-1} \cup \{i_1, \ldots, i_{j-1}, a\} \setminus \{c_1,\cdots, c_{j-1}, \pi(a) \} |T_{1,\ldots, w-1}, i_{1,\ldots, j-1}, a\in S^*\cap Z_{s_j}] 
\\\ge&
\Ex[f({S}_{1,\ldots, w-1} \cup \{i_1, \ldots, i_{j-1}\} \setminus \{c_1,\cdots, c_{j-1}, \pi(a) \} )|T_{1,\ldots, w-1}, i_{1,\ldots, j-1}, a\in S^*\cap Z_{s_j}]
\\&+\Ex[\Delta_f(a|S_{1,\cdots,w-1}\cup\{i_1,\cdots,i_{j-1}\}\setminus \{c_1,\cdots, c_{j-1},\pi(a)\} )|T_{1,\ldots, w-1}, i_{1,\ldots, j-1}, a\in S^*\cap Z_{s_j}] \\
\ge&
\Ex[f({S}_{1,\ldots, w-1} \cup \{i_1, \ldots, i_{j-1}\} \setminus \{c_1,\cdots, c_{j-1}, \pi(a) \} )|T_{1,\ldots, w-1}, i_{1,\ldots, j-1}, a\in S^*\cap Z_{s_j}]
\\&+\Ex[\Delta_f(a|S_{1,\cdots,w-1}\cup\{i_1,\cdots,i_{j-1}\}\setminus \{c_1,\cdots, c_{j-1}\} )|T_{1,\ldots, w-1}, i_{1,\ldots, j-1}, a\in S^*\cap Z_{s_j}]
\end{eqnarray*}
The first inequality is from the definition of function $g$ as it is defined in equation~\ref{eq:g}. The last inequality from submoularity of $f$.
Now from the last inequality  and lemma~\ref{replaced} we have
\begin{align*}
&\Ex[f({S}_{1,\ldots, w-1} \cup \{i_1, \ldots, i_{j}\} \setminus \{c_1,\cdots, c_j \} |T_{1,\ldots, w-1}, i_{1,\ldots, j-1})]& \\ &\ge
(1-\frac{1}{k-\alpha} ) f(S_{1,\cdots, w-1}\cup \{i_1, \ldots, i_{j-1}\} \setminus \{c_1, \cdots, c_{j-1}\} )&\\
&+ \Ex[ \Delta_f(a|{S}_{1,\ldots, w-1} \cup \{i_1, \ldots, i_{j-1}\} \setminus\{c_1,\cdots, c_{j-1}\} \})||T_{1,\ldots, w-1}, i_{1,\ldots, j-1}, a\in S^*\cap Z_{s_j}] &
\end{align*}
Now from lemma~\ref{lem:asGoodas} and the above inequality we can show
\begin{align*}
&\Ex[f({S}_{1,\ldots, w-1} \cup \{i_1, \ldots, i_{j}\} \setminus \{c_1,\cdots, c_j \} |T_{1,\ldots, w-1}, i_{1,\ldots, j-1})]& \\  &\ge
(1-\frac{1}{k-\alpha} ) f(S_{1,\cdots, w-1}\cup \{i_1, \ldots, i_{j-1}\} \setminus \{c_1, \cdots, c_{j-1}\} )&\\
&+{ \frac{1}{k}\left((1-\frac{\alpha}{k})f(S^*)-f({S}_{1,\ldots, w-1} \cup \{i_1, \ldots, i_{j-1}\} \setminus\{c_1,\cdots, c_{j-1}\})\right)\ . }&
\end{align*}
Thus,
\begin{align}\label{recursion}
f({S}_{1,\ldots, w-1} \cup \{i_1, \ldots, i_{j}\} \setminus \{c_1,\cdots, c_j \} )  
- f(S_{1,\cdots, w-1}\cup \{i_1, \ldots, i_{j-1}\} \setminus \{c_1, \cdots, c_{j-1}\} )\\
\ge
{ \frac{1}{k}\left((1-\frac{\alpha}{k})f(S^*)-2f({S}_{1,\ldots, w-1} \cup \{i_1,\cdots,i_{j-1}\}\setminus \{c_1,\cdots, c_{j-1}\})\right)\ . }
\end{align}

\end{proof}

Using standard techniques for the analysis of greedy algorithm, the following corollary of the previous lemma can be derived, 
\begin{lemma}
\label{cor:asGoodas}
$$\Ex\left[\left(1-\frac{\alpha}{k}\right) f(S^*) - 2f(S_{1,\ldots, w-1}\cup \gamma(\tilde \tau_w)\setminus \zeta(\tilde{\tau}_{w})) | T\right]\le \Ex\left[e^{-\frac{2|\tilde \tau_w|}{k}} \left|\right. T\right] \left(\left(1-\frac{\alpha}{k}\right) f(S^*)-2f(S_{1,\ldots, w-1})\right)$$
\end{lemma}
\begin{proof}
Let $\pi_0= (1-\frac{\alpha}{k}) f(S^*) - 2\Ex[f(S_{1,\ldots, w-1}) | T_{1,\ldots, w-1}=T]$, and for $j \ge 1$,
$$\pi_j:=(1-\frac{\alpha}{k}) f(S^*) - 2\Ex[f(S_{1,\ldots, w-1}\cup \{i_1, \ldots, i_j\}\setminus\{c_1,\cdots,c_j\}) | T_{1,\ldots, w-1}=T, i_1,\ldots, i_{j-1}],$$
Then, subtracting and adding $\frac{1}{2}(1-\frac{\alpha}{k}) f(S^*)$ from the left hand side of lemma~\ref{exg}, and taking expectation conditional on $T_{1,\ldots, w-1}=T, i_1, \ldots, i_{j-2}$, we get
$$-\frac{1}{2} ( \Ex[\pi_{j} | T, i_1, \ldots, i_{j-2}] + \pi_{j-1} ) \ge \frac{1}{k} \pi_{j-1}$$
which implies
$$\Ex[\pi_j|T, i_1, \ldots, i_{j-2}] \le \left(1-\frac{2}{k}\right) \pi_{j-1} \le \left(1-\frac{2}{k}\right)^j \pi_0\ .$$
By martingale stopping theorem, this implies:
$$\Ex[\pi_t|T] \le \Ex\left[\left(1-\frac{2}{k}\right)^t \left| T\right. \right] \pi_0 \le \Ex\left[e^{-2t/k}| T\right] \pi_0\ .$$
where stopping time $t=|\tilde{\tau}_w|$. ($t=|\tilde \tau_w| \le \alpha\beta$ is bounded, therefore, martingale stopping theorem can be applied).

\end{proof}

Next, we compare $\gamma(\tilde \tau_w)$ to $S_w=\gamma(\tau^*)$ . Here, $\tau^*$ was defined has the `best' greedy subsequence of length $\alpha$ (refer to \eqref{eq:Sw} and \eqref{eq:taustar}). To compare it with $\tilde \tau_w$, we need a bound on size of $\tilde \tau_w$. We use concentration inequalities proved in~\cite{us}:

\begin{lemma}
[proved in~\cite{us}]
\label{cor:lengthtau}
For any real $\delta'\in (0,1)$,  if parameters $k,\alpha, \beta$ satisfy \settinga, then given any $T_{1,\ldots, w-1}=T$, with probability at least $1-\delta' e^{-\alpha/k}$,
$$|\tilde \tau_w| \ge (1-\delta') \alpha\ . $$

\end{lemma}
\toRemove{
\begin{proof}
We use the previous lemma with $\delta=\delta'/2$ to get lower bound of $(1-\delta')\alpha$ with probability $1-\exp(-(\delta')^2\alpha/32\beta)$. Then, substituting 
$ k\ge \alpha\beta \ge \frac{64\beta}{(\delta')^2} \log(1/\delta')$ so that using $\beta \le \frac{k(\delta')^2}{64 \log(1/\delta')}$ we can bound the violation probability by
$$\exp(-(\delta')^2 \alpha/32\beta)\le \exp(-(\delta')^2 \alpha/64\beta)\exp(-\alpha/k) \le \delta' e^{-\alpha/k}.$$
where the last inequality uses $\alpha\ge 8\beta^2 \log(1/\delta')$ and $\beta \ge 8/(\delta')^2$.
\end{proof}
}

\begin{lemma}
\label{lem:Sw}
For any real $\delta'\in (0,1)$, if parameters 
$k,\alpha, \beta$ satisfy \settinga, then
$$\Ex\left[\frac{k-\alpha}{k} \OPT -2f(S_{1,\ldots, w})|T_{1,\ldots, w-1}\right] \le (1-\delta') e^{-2\alpha/k} \left(\frac{k-\alpha}{k} \OPT - 2f(S_{1,\ldots, w-1})\right)\ .$$
\end{lemma}
\begin{proof}
The lemma follows from substituting Lemma \ref{cor:lengthtau} in Lemma \ref{cor:asGoodas}.
\end{proof}

\matroidThm

\begin{proof}

Now from Lemma~\ref{lem:Sw}, we have,
for any real $\delta'\in (0,1)$, if parameters 
$k,\alpha, \beta$ satisfy \settinga, then the set $S_{1,\ldots, W}$ tracked by Algorithm \ref{alg:main} satisfies
$$\mathbb{E}[f(S_{1,\ldots, W})] \ge (1-\delta')^2 (\frac{1}{2}(1-1/e^2) ) \OPT.$$
\toRemove{
\begin{proof}
By multiplying the inequality Lemma \ref{lem:Sw} from $w=1, \ldots, W$, where $W=k/\alpha$, we get 
$$\mathbb{E}[f(S_{1,\ldots, W})] \ge (1-\delta')(\frac{1}{2}(1-1/e^2)) (1-\frac{\alpha}{k}) \OPT.$$
Then, using $1-\frac{\alpha}{k}\ge 1-\delta'$ because $k\ge \alpha \beta \ge \frac{\alpha}{\delta'}$, we obtain the desired statement.
\end{proof}
}



Now, we compare $f(S_{1\ldots, W})$ to $f(A^*)$, where $A^*=S_{1\ldots, W} \cap A$, with $A$ being the shortlist returned by Algorithm \ref{alg:main}. The main difference between the two sets is that in construction of shortlist $A$, Algorithm \ref{alg:matroidmax} is being used to compute the argmax in the definition of $\gamma(\tau)$, in an online manner. This argmax may not be computed exactly, so that some items from $S_{1\ldots, W}$ may not be part of the shortlist $A$.

Similar to Lemma 16 in~\cite{us}, we can show that each element in $A$ gets selected by the algorithm with probability at least $1-\delta$. More precisely,
let $A$ be the shortlist returned by Algorithm \ref{alg:main}, and $\delta$ is the parameter used to call Algorithm \ref{alg:matroidmax} in Algorithm \ref{alg:main}. Then, for given configuration $Y$, for any item $a$, we have $$Pr(a\in A|Y, a\in S_{1,\cdots, w}) \ge 1-\delta\ .$$
Therefore using Lemma~\ref{sample},
\label{prop:online}
$$\Ex[f(A^*)] := \mathbb{E}[f(S_{1,\cdots, W} \cap A)] \ge (1-\frac{\epsilon}{2})\mathbb{E}[f(S_{1,\cdots, W})]$$
where $A^*:= S_{1,\cdots, W} \cap A$ is the 
subset of shortlist $A$ returned by Algorithm \ref{alg:main}.
The proof is similar to the proof in~\cite{us}.

\end{proof}

\subsection{Preemption model and  Shorlitst of size at most $k$}
Finally we focus on the special case where the size of shortlist is at most $k$. We can get a constant competitive algorithm even with the slight relaxation of the \textit{matroid secretary problem} 
to the case that we allow the algorithm to select a shortlist of size at most $k=rk(\mathcal{M})$. 
The algorithm finally outputs an independent subset of this shortlist of size $k$. There was no constant compettetive algorithm even for this natural relaxation of \textit{matroid secretary problem}.
Also we are not aware of any direct way to prove a constant factor guarantee for this simple relaxation without using the techniques that we develop using $(\alpha,\beta)$-windows.

\thmpreemption

\begin{proof}
We show that  algorithm~\ref{alg:main} 
with parameter $\alpha=\beta=1$ satisfies the above mentioned properties. Firstly,  algorithm~\ref{alg:main} (with $\alpha=1$, and $\beta=1$) uses shortlist of size $\eta(k)\le k$.
The reason is that the algorithm divides the input into exactly $k$ slots.
Also each window contains exactly one slot. 
The function $\gamma$ tries all $\alpha$-subsequences of a window which is exactly one slot. Thus $\gamma$ returns one element in that slot with hight value of $g(e,S)$ as defined in~\ref{eq:g}, which might cause removal of at most one element $\theta(S,e)$ from the current solution $S$. Therefore the algorithm has shortlist size at most $k$ and also satisfies the preemption model. Now by setting $\alpha=1, \beta=1$ we can get a constant compettetive ratio that  the error rate comes from lemma~\ref{lem:Sw}.

\toRemove{
Given that $|OPT|\le k$, we can extend the size of $OPT$ to a set of size exactly $k$, by adding some dummy elements.
Let's call that set $OPT'$.
In order to calculate the competitive ratio, we first compute the probability that one slot contains an element of optimal solution $OPT'$, i.e., the probability that $s\cap OPT'\neq \emptyset$, for a slot $s$ in window $w$. Since all $k$ elements of $OPT'$ are uniformly distributed, 
$$
Pr[s\cap OPT' \neq \emptyset | T_{1,\cdots, w-1} ] \ge 1- (1-1/k)^k.
$$
Therefore in lemma~\ref{cor:lengthtau}, $|\bar{\tau}_{w}|\ge 1$ with probability $1-1/e$. Hence 
$\mathbb{E}[f(S_{1,\ldots, W})] \ge \frac{1}{2}(1-1/e) (1-1/e^2)  \OPT.$
Thus by Lemma~\ref{cor:lengthtau} and~\ref{lem:Sw}, 
$$\mathbb{E}[f(A^*)] \ge  \frac{1}{2} (1-\epsilon)(1-1/e) (1-1/e^2)  \OPT.$$
}
\end{proof}

\section{ $p$-matchoid constraints}
\label{sec:matchoid}
In this section, we present  algorithms
for monotone submodular function maximization subject to $p$-matchoid constraints.
These constraints generalize many basic combinatorial constraints such as the cardinality constraint, the intersection of $p$ matroids, and matchings in graphs. 
Throughout this section,  $k$ would refer to the size of the largest feasible set.
A formal definition of
a $p$-matchoid is as follows:

\begin{definition} \label{def:matchoid}
(\textbf{Matchoids}). Let $\mathcal{M}_1 = (\mathcal{N}_1, \mathcal{I}_1), \cdots ,\mathcal{M}_q = (\mathcal{N}_q, \mathcal{I}_q)$ be $q$ matroids over overlapping groundsets. Let $\mathcal{N} = \mathcal{N}_1\cup  \cdots \cup \mathcal{N}_q$ and $\mathcal{I} = \{S\subseteq \mathcal{N} : S\cap \mathcal{N} \in \mathcal{I}_{\ell},
 \forall \ell\}$. The finite set system
$\mathcal{M}_p = (\mathcal{N} , \mathcal{I})$ is a $p$-matchoid if for every element $e \in \mathcal{N}$ , $e$ is a member of $\mathcal{N}$ for at most $p$ indices $\ell \in [q]$.
\end{definition}


There are some subtle differences in the algorithm as well as in the analysis.
The main difference in the algorithm is that instead of removing one item from the current independent set $S$, we might remove up to $p$ items form  $S$. Each removed item corresponds to different ground sett $N_i$, in which the new item lies (based on the definition of $p$-matchoid constraints, Definition~\ref{def:matchoid}, there are at most $p$ such elements).

For each index $\ell\in [q]$  define:
\begin{equation}
\Omega_{\ell}(e,S):= \{e'\in S| S+e-e' \in \mathcal{I}_{\ell} \} 
\end{equation}
For an element $e$ in the input, suppose $e\in N_{\ell_i}$, for $i=1,\cdots, p$.
Define 
\begin{equation}
\lambda(e,S):= \prod_{i=1}^{p} {\Omega_{\ell_i}(e,S)}
\end{equation}
For a combination vector 
$r=(r_1,\cdots, r_p)\in \lambda(e,S)$, where
$r_i \in \Omega_{\ell_i}(e,S)$,
define:
\begin{equation}
\mu(r):=  \{r_1,\cdots, r_p\}
\end{equation}
\begin{equation}
g_r(e,S):= f(S+e-\mu(r))-f(S)
\end{equation}
Also define:
\begin{equation}
\theta(e,S):= \mu( \arg\max_{r\in \lambda(e,S) } g_r(e,S))
\end{equation}
Furthermore define,
\begin{equation}
\label{eq:newg}
g(e,S):= \max_{r\in \lambda(e,S) } g_r(e,S)
\end{equation}

\toRemove{
For each index $\ell\in [q]$  define the following functions as before:
$$ \theta_{\ell}(e,S) := \arg\max_{e'\in S} \{ f(S+e-e')| S+e-e' \in \mathcal{I}_{\ell} \} $$
Also define $g_{\ell}$ as 
 $$g_{\ell}(e,S):= f(S+e-\theta_{\ell}(e,S)) - f(S)$$
and
$$g(e,S):= f(S+e-\bigcup_{i}\theta_i(e,S)) - f(S)$$
 }
As in the online subroutine for the main algorithm, we run Algorithm~\ref{alg:matroidmax} with the new function $g$ defined in equation~\ref{eq:newg}. It returns  element  $e$ with maximum  $g(e,S)$, and it achieves a  $1-\delta$ competitive ratio   with shortlists of size logarithmic in $1/\delta$.

Additionally, we  make some changes in the main algorithm~\ref{alg:main}. In particular, we define $\gamma$ similar to equation~\ref{eq:gamma} but using the new definition of $g$ in equation~\ref{eq:newg}.
Moreover, for a subsequence $\tau=(s_1,\ldots, s_\ell)$
define 
\begin{equation}
\zeta(\tau):=\bigcup_{j=1}^{\ell} C_j
\end{equation}
where each $C_j$ is a set defined as
\begin{equation}
\label{eq:cij}
C_j :=  
\theta(i_j, S_{1, \ldots,w-1} \cup \{i_1,\ldots, i_{j-1}\})
\end{equation}
Note that in contrast with the definition of $\zeta(\tau)$ for the matroid constraints equation~\ref{eq:cij}, in which $c_j$ is only one item, now each $C_j$ is a subset of the current independent set $S$. Further, the definition of $\bar{S}_w$, in equation~\ref{eq:Swbar}, will be updated accordingly using the new definition of $\zeta(\tau)$.

Now we can generalize Lemma~\ref{replaced} to $p$-matchoid constraints.

\begin{lemma}\label{replacedMatchoid}
Suppose the sequence $\tilde \tau_w=(s_1, \ldots, s_t)$ defined as in Definition \ref{def:tauw}, let $\gamma(\tilde \tau_s)=(i_1,\ldots, i_t)$, with $\gamma(\cdot)$ as defined in \eqref{eq:gamma}. 
For any $j\in \{1,\ldots, t\}$, and element $b\in \mathcal{N}_{\ell}$, 
let $S'_{\ell}$ be the extension of $S_{1,\cdots, w-1}\cup \{i_1, \cdots, i_{j-1}\}\setminus\bigcup_{r\le j-1} C_r$ to an independent set in $\mathcal{M}_{\ell}$, and $\pi_{\ell}$ be the bijection from Brualdi lemma (refer to Lemma~\ref{lem:Brualdi}) from $S^*$ to $S'_{\ell}$. 
Further, let's denote by $\pi(b):=\{\pi_{\ell}(b)|b \in \mathcal{N}_{\ell}\}$, then
{\small \begin{eqnarray*}
\Ex[f(S_{1,\ldots, w-1} \cup \{i_1, \ldots, i_{j-1}\}\setminus (\bigcup_{r\le j-1} C_r \cup \pi(a)|T_{1,\ldots, w-1}, i_{1,\ldots, j-1}, a\in S^*\cap Z_{s_j}]\\
 \ge (1-\frac{p}{k})f(S_{1,\ldots, w-1} \cup \{i_1, \ldots, i_{j-1}\}\setminus\{c_1,\cdots,c_{j-1}\})\
\end{eqnarray*}}
\end{lemma}
\begin{proof}

The proof is similar to the proof of Lemma~\ref{replaced}. 
For $\ell \in [q]$,
since $\pi_{\ell}$ is a bijection from $S^*\cap \mathcal{N}_{\ell}$ to $S'_{\ell}$, from Brualdi's lemma (lemma~\ref{lem:Brualdi}), there is an onto mapping $\pi'_{\ell}$ from $S^*\cap \mathcal{N}_{\ell}$ 
 to $S_{1,\cdots, w-1}\cup \{i_1, \cdots, i_{j-1}\}\setminus (\bigcup_{r\le j-1} C_r \cup \pi(a))\cup \{\emptyset\}$ 
such that 
$S_{1,\cdots, w-1}\cup \{i_1, \cdots, i_{j-1}\}\setminus (\bigcup_{r\le j-1} C_r \cup \pi(a))- \pi'_{\ell}(a) +a \in M_{\ell}$, for all $a \in S^*$. Further, $\pi'_{\ell}(a)=\pi_{\ell}(a)$ if $\pi_{\ell}(a)\in  S_{1,\cdots, w-1}\cup \{i_1, \cdots, i_{j-1}\}\setminus \bigcup_{r\le j-1} C_r$ and $\pi'_{\ell}(a)=\emptyset$ otherwise.

Recall the definition of $Z_{s_j}$ (refer to definition~\ref{def:tauw}).
Suppose $a$ is a randomly picked item from $S^*\cap Z_{s_j}$.
Note that from Lemma~\ref{lem:Zs}, conditioned on $T_{1,\cdots, w-1}$, the element $a$ can be equally any element of 
$S^*\backslash  \{Z_1,\ldots Z_{s_{j-1}}\}$
with probability at least $1/k$.
Therefore, $\pi'_{\ell}(a)$ would be any of $S_{1,\cdots, w-1}\cup \{i_1, \cdots, i_{j-1}\}\setminus \bigcup_{r\le j-1} C_r $ with probability at most $1/k$ 
(since $\pi'_{\ell}$ might map some elements of $S^*$ to the empty set).

For  element $e\in S_{1,\cdots, w-1}\cup \{i_1, \cdots, i_{j-1}\}\setminus \bigcup_{r\le j-1} C_r$, let $\mathcal{N}(e)$ be the set of indices $\ell$ such that $e\in \mathcal{N}_{\ell}$. Because of the $p$-matchoid constraint, we have $|\mathcal{N}(e)|\le p$.
Define 
$$
\pi^{-1}(e):=\{t | t\in \mathcal{N}_{\ell}, \text{ for some }  \ell\in \mathcal{N}(e) \text{ and }  \pi_{\ell}(t)=e \} 
$$
we have also $|\pi^{-1}(e)|\le p$. Thus, each element $e \in S_{1,\cdots, w-1}\cup \{i_1, \cdots, i_{j-1}\}\setminus \bigcup_{r\le j-1} C_r$ belongs to 
$\pi(a)$ with probability at most $p/k$:
$$
\Pr(e\in \pi(a) | a\in S^{*}\cap Z_{s_j} ) = \Pr(a\in S^{*}\cap Z_{s_j}\cap \pi^{-1}(a) ) \le \frac{p}{k} 
$$
 Now we apply Lemma~\ref{sample}.
It is crucial to note that in Lemma~\ref{sample}  each element do not need to be selected necessarily independently.
Definition of $\pi$ and lemma~\ref{sample} imply the lemma.
\end{proof}

Furthermore the main difference in the analysis is that instead of recursion~\ref{recursion}, we get the following new recursion:
\toRemove{
\begin{align}\label{matchoidRecursion}
f({S}_{1,\ldots, w-1} \cup \{i_1, \ldots, i_{j}\} \setminus \bigcup_{r\le j} C_r )  
- f(S_{1,\cdots, w-1}\cup \{i_1, \ldots, i_{j-1}\} \setminus \bigcup_{r\le j-1} C_r )\\
\ge
{ \frac{1}{k}\left((1-\frac{\alpha}{k})f(S^*)-pf({S}_{1,\ldots, w-1} \cup \{i_1, \ldots, i_{j-1}\})\right)\ . }
\end{align}
}

\begin{lemma}\label{exgmatchoid}
Suppose function $g$ is as defined in equation~\ref{eq:newg}.
For the sequence $\tilde \tau_w=(s_1, \ldots, s_t)$, and $\gamma(\tilde \tau_s)=(i_1,\ldots, i_t)$. Then, for all $j=1,\ldots, t$, \\
{\small \begin{eqnarray*}
\Ex \left[g(i_j,S_{1,\ldots, w-1} \cup \{i_1, \ldots, i_{j-1}\}\setminus \bigcup_{r\le j-1} C_r)|T_{1,\ldots, w-1}, i_{1,\ldots, j-1}\right]\\
 \ge 
 \frac{1}{k}
 \left(
 (1-\frac{\alpha}{k})f(S^*)-(p+1) f(S_{1,\ldots, w-1} \cup \{i_1, \ldots, i_{j-1}\}\setminus\bigcup_{r\le j-1} C_r)\right)\
\end{eqnarray*}}
\end{lemma}


\begin{proof}
The proof is similar to the proof of Lemma~\ref{exg} with some changes regarding matchoid constraints.
In the algorithm~\ref{alg:main}, at the end of window $w$, we set 
$S_{1,\cdots, w} = S_{1,\cdots, w-1}\cup S_w \setminus \bar{S}_w$.
Suppose $a\in s_j\cap S^*$. 
Moreover, let $S'_{\ell}$ be the extension of $S_{1,\cdots, w-1}\cup \{i_1, \cdots, i_{j-1}\}\setminus \bigcup_{r\le j-1} C_r$ to an independent set in $\mathcal{M}_{\ell}$, and $\pi_{\ell}$ be the bijection in Brualdi lemma (refer to Lemma~\ref{lem:Brualdi}) from $S^*_{\ell}$ to $S'_{\ell}$.
Thus the expected value of the function $g$ on the element selected by the algorithm in slot $s_j$ (the element with maximum $g$ in the slot $s_j$) would be
\begin{eqnarray*}
&\Ex[f({S}_{1,\ldots, w-1} \cup \{i_1, \ldots, i_{j}\} \setminus \bigcup_{r\le j} C_r)|T_{1,\ldots, w-1}, i_{1,\ldots, j-1}]  
\\\ge&
\Ex[f({S}_{1,\ldots, w-1} \cup \{i_1, \ldots, i_{j-1}, a\} \setminus \bigcup_{r\le j-1} C_r \cup \pi(a)  |T_{1,\ldots, w-1}, i_{1,\ldots, j-1}, a\in S^*\cap Z_{s_j}] 
\\\ge&
\Ex[f({S}_{1,\ldots, w-1} \cup \{i_1, \ldots, i_{j-1}\} \setminus \bigcup_{r\le j-1} C_r \cup \pi(a) )|T_{1,\ldots, w-1}, i_{1,\ldots, j-1}, a\in S^*\cap Z_{s_j}]
\\&+\Ex[\Delta_f(a|S_{1,\cdots,w-1}\cup\{i_1,\cdots,i_{j-1}\}\setminus \bigcup_{r\le j-1} C_r \cup \pi(a) )|T_{1,\ldots, w-1}, i_{1,\ldots, j-1}, a\in S^*\cap Z_{s_j}] \\
\ge&
\Ex[f({S}_{1,\ldots, w-1} \cup \{i_1, \ldots, i_{j-1}\} \setminus \bigcup_{r\le j-1} C_r \cup  \pi(a) )|T_{1,\ldots, w-1}, i_{1,\ldots, j-1}, a\in S^*\cap Z_{s_j}]
\\&+\Ex[\Delta_f(a|S_{1,\cdots,w-1}\cup\{i_1,\cdots,i_{j-1}\}\setminus \bigcup_{r\le j-1} C_r )|T_{1,\ldots, w-1}, i_{1,\ldots, j-1}, a\in S^*\cap Z_{s_j}]
\end{eqnarray*}
The first inequality is from the definition of function $g$ as it is defined in equation~\ref{eq:newg} and the fact that the algrotihm selects an element in slot $s_j$ with maximum value of $g$. The second inequality is from submodularity and the last inequality is from monotonicity of $f$.
Now from the last inequality  and Lemma~\ref{exgmatchoid}, we can show,
\begin{align*}
&\Ex[f({S}_{1,\ldots, w-1} \cup \{i_1, \ldots, i_{j}\} \setminus \bigcup_{r\le j} C_r) |T_{1,\ldots, w-1}, i_{1,\ldots, j-1}]& \\ &\ge
(1-\frac{p}{k} ) f(S_{1,\cdots, w-1}\cup \{i_1, \ldots, i_{j-1}\} \setminus \bigcup_{r\le j-1} C_r )&\\
&+ \Ex[ \Delta_f(a|{S}_{1,\ldots, w-1} \cup \{i_1, \ldots, i_{j-1}\} \setminus\bigcup_{r\le j-1} C_r )||T_{1,\ldots, w-1}, i_{1,\ldots, j-1}, a\in S^*\cap Z_{s_j}] &
\end{align*}
Now from lemma~\ref{lem:asGoodas} and the above inequality we can show
\begin{align*}
&\Ex[f({S}_{1,\ldots, w-1} \cup \{i_1, \ldots, i_{j}\} \setminus \bigcup_{r\le j} C_r) |T_{1,\ldots, w-1}, i_{1,\ldots, j-1}]& \\  &\ge
(1-\frac{p}{k} ) f(S_{1,\cdots, w-1}\cup \{i_1, \ldots, i_{j-1}\} \setminus \bigcup_{r\le j-1} C_r )&\\
&+ { \frac{1}{k}\left((1-\frac{\alpha}{k})f(S^*)-f({S}_{1,\ldots, w-1} \cup \{i_1, \ldots, i_{j-1}\} \setminus\bigcup_{r\le j-1} C_r ) \right)\ . }&
\end{align*}
Thus,
\begin{align}\label{newrecursion}
f({S}_{1,\ldots, w-1} \cup \{i_1, \ldots, i_{j}\} \setminus \bigcup_{r\le j} C_r )  
- f(S_{1,\cdots, w-1}\cup \{i_1, \ldots, i_{j-1}\} \setminus \bigcup_{r\le j-1} C_r )\\
\ge
{ \frac{1}{k}\left((1-\frac{\alpha}{k})f(S^*)-(p+1)f({S}_{1,\ldots, w-1} \cup \{i_1,\cdots,i_{j-1}\}\setminus \bigcup_{r\le j-1}C_r)\right)\ . }
\end{align}

\end{proof}

By solving the recursion and similar to the analysis for matroid constraints we can show the following theorem:

\matchoidThm

\section{Streaming 
}
\label{sec:streaming}
In this section, we show that Algorithm \ref{alg:main} can be implemented in a way that it uses a memory buffer of size at most $\eta(k)=O(k)$; also we compute the number of objective function evaluations for each arriving item as follows. 

\thmStreamingMatroid

Similarly for $p$-matchoid constraint we have the following result for the streaming setting:

\thmStreamingMatchoid
\begin{proof}
Th difference between Algorithm~\ref{alg:main} in this paper and the main Algorithm  in~\cite{us} is that, we remove elements of $\bar{S}_w$  from $S$ at the end of each window $w$. Therefore, with the same argument in the proof of Theorem 2 in~\cite{us}, we keep track of all parameters in Algorithm~\ref{alg:main} including $\bar{S}_w, S_w, R_w, \hat{S}_w$ in a memory efficient way using memory $O(k)$. 
The other difference between the two algorithms is in the subroutine~\ref{alg:matroidmax} that finds the element with maximum $g$ in a slot. 
In~\cite{us}, $g(e,S)$ can be computed using only one oracle access, whereas in the new definition of $g$ in equation~\ref{eq:newg}, we need access to independence oracle of $p$ matroids that $e$ belongs to, in order to check the independence of $S+e-e'$ for each $e' \in S$. At most $p \kappa$ elements of $S$ are eligible (they are in the ground set of a matroid that $e$ also member of). Hence, 
in order to create $\Omega_{\ell}(e,S)$, 
 for each arriving element $e$ in the input, we need $O(p\kappa)$ access to Independence oracle. Similarly the total access to the value oracle is $O(p\kappa)$. 
 In order to compute 
$\lambda(e,S)$, we need to consider all $\kappa^p$ combinations and have access to value oracle. Therefore the  number of access to the value oracle is $O(p\kappa+\kappa^p)$ per element. 
But, since the first element $a_0$ is computed in the beginning of each slot
for each $\tau$, we would have in average an additional $O(k^2/n)$ function evaluation per element.
\end{proof}

\toRemove{
\subsection{Max-Coverage}
One important monotone submodular function is max-coverage. 
\begin{definition} (max-coverage)
\end{definition}

Oracle access 

memory required to store each set.
}

 In the next section, we empirically compare our streaming algorithms with the state of the art algorithms in the streaming setting.

\subsection{Experiment}
In this section, we consider different types of constraints including uniform matroid, intersection of partition matroids and $p$-matchoid constraints. We compare our algorithm with state of the art algorithm for each type of constraint using YouTube dataset and Twitter dataset described in the next section.

\subsubsection{DataSets}
The experiments will be on a Twitter stream summarization task and a YouTube Video
summarization task similar to the one in Kazemi et al.~\cite{kazemi2019submodular}.

\paragraph{\textbf{Twitter Stream Summarization}} In this application, we want to produce real-time summaries
for Twitter feeds. 
It is valuable to create
a succinct summary that contains all the important information.
We use the dataset created by~\cite{kazemi2019submodular}.
They gather recent tweets from 30 different popular news accounts, to collect
a total of 42,104 unique tweets. 
They also define a monotone submodular function $f$ 
that measure the redundancy of important stories in a set $S$.
It is defined as follows on a set $S\subseteq V$ of tweets:
$$
f(S):=\sum_{w\in W} \sqrt{\sum_{e \in S} score(w,e)}
$$

function f defined over a ground set $V$ of tweets. Each tweet $e\in V$ consists of a positive value
vale denoting its number of retweets and a set of $\ell_e$ keywords $W_e = \{w_{e,1}, \cdots, w_{e,\ell_e}\}$
from a general
set of keywords $W$. The score of a word $w\in W_e$ for a tweet $e$ is defined by $score(w, e) = vale_e$. If
$w\notin W_e$. Define $score(w, e) = 0$.

\paragraph{\textbf{YouTube Video Summarization}}
For the YouTube dataset, we want to select a subset of frames from video feeds which are representative of the entire video.
We use the same dataset as in~\cite{kazemi2019submodular}, which is YouTube videos of
New Year’s Eve celebrations from ten different cities around the world.

They compresses each frame into
a 4-dimensional representative vector. Given a ground set $V$ of such vectors,  define a matrix $M$
such that $M_{ij}=e-dist(v_i,v_j)$
, where $dist(v_i, v_j )$ is the euclidean distance between vectors 
$v_i, v_j \in V$.
Intuitively, $M_{ij}$ encodes the similarity between the frames represented by $v_i$ and $v_j$. 
They define a function  that intuitively measure the diversity of the
vectors in a set $S$ as follows: $f(S) =
\log det(I + \alpha M_S)$, where $I$ is the identity matrix, $\alpha > 0$ and $M_S$ is the principal sub-matrix of $M$
indexed by $S$. 


\subsubsection{\textbf{Uniform Matroid}}
The simplest constraint that we can impose is the uniform matroid or equivalently the cardinality constraint.
In the simplest form our algorithm is similar to~\cite{us}.
We compare our algorithm to the state of the art algorithm in the streaming setting~\cite{kazemi2019submodular}.
As we established
an upper bound on the constant factor $\eta_{\epsilon}(k)$ in theorem~\ref{opttheoremmatchoid}, the performance of our algorithm crucially depends on the choice of 
$\alpha$ and $\beta$. The running time also is a function of $\alpha$ and $\beta$, and it grows rapidly as we increase $\alpha$ and $\beta$. 
Surprisingly, our algorithm outperforms~\cite{kazemi2019submodular}
substantially even with relatively small choices of $\alpha=6$ and $\beta=2$.
We also observe that the utility of the output returned by our algorithm can be very  close to  what the optimal offline algorithm, namely the Greedy algorithm achieves. In Figure~\ref{fig:uniform}, we have plotted the performance of all three algorithms on the YouTube dataset.
Note that in our experiment we use a simplistic version of our algorithm in which we subsample from the shortlist in beginning of each window and only use that subsample rather than the entire shortlist. Furthermore we observe that our algorithm is slower than~\cite{kazemi2019submodular}, but the interesting fact about our algorithm as stated in Theorem~\ref{opttheoremmatchoid} is that it is highly parallel thus it has the potential to become $\eta_{\epsilon}(k)$  times faster.
\begin{figure}[ht!]
\centering
\includegraphics[width=120mm]{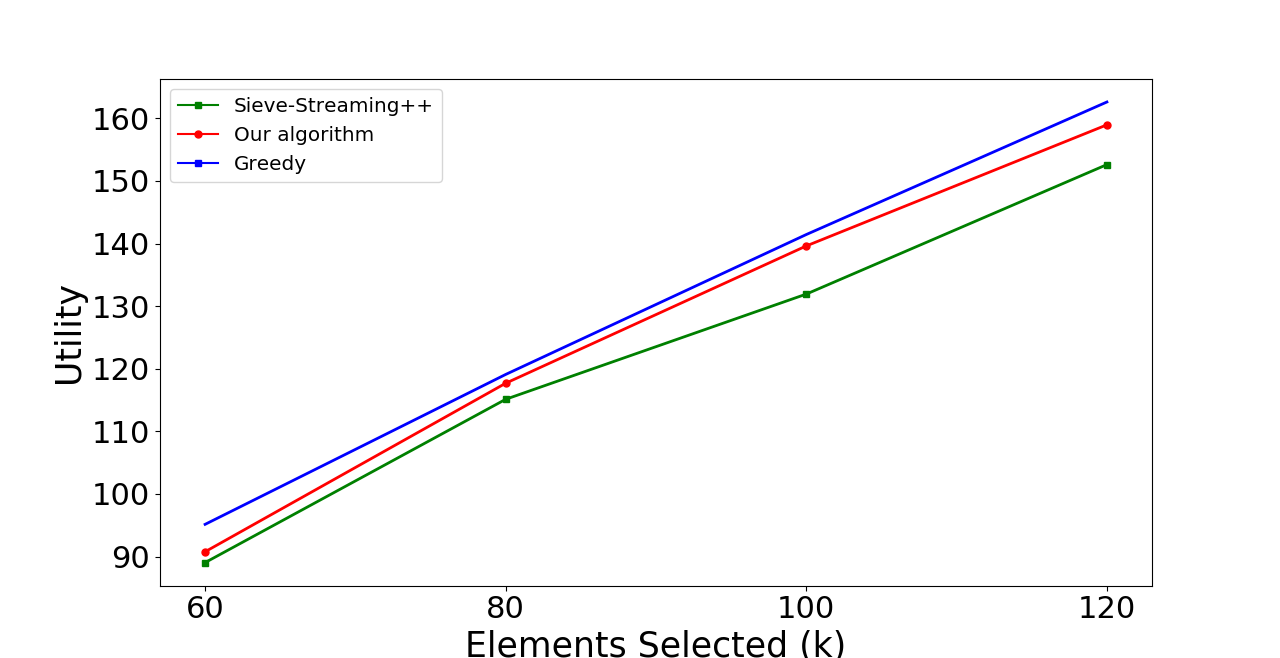}
\caption{The plot is for unifrom matroid, $\alpha =6$ and $\beta=2$}
\label{fig:uniform}
\end{figure}

\subsubsection{\textbf{$p$-matchoid constraints}}
For the case of $p$-matchoid constraints, state of the art algorithm for general streaming setting is due to Feldman et al.~\cite{feldman2018streaming}.
In our experiment, we divide the elements of input into $q$ categories $\mathcal{N}= \mathcal{N}_1 \cup \cdots, \cup  \mathcal{N}_q$. 
We assign $p$ tags to each element $e$.
Each tag belongs to one of the catergories $1,\cdots, q$ (generated randomly). 
Further, we impose a cardinality constraint $3$ for each category (i.e, $\mathcal{I}_{\ell} $ is a cardinality constraint).
The objective is to select at most $3$ elements from each category. In other words, an independent set of $p$-matchoid is defined as 
$$\mathcal{I}=\{S\subseteq \mathcal{N}: |S\cap \mathcal{N}_i|\le 3, \forall i\in [q] \}$$
In our algorithm, we  set $\alpha=3$ and $\beta=2$. We have plotted the performance of our algorithms and~\cite{feldman2018streaming} on the Twitter dataset below.
The first plot, Figure~\ref{fig:p3}, is for fixed $p=3$ and different number of categories $q$.
The second plot, Figure~\ref{fig:pmatchoid}, is for fixed number of categories $q=30$ and different values of $p$ from $1,\cdots, 10$. As the competitive ratio of our algorithm suggests, by  increasing $p$ the ratio of our utility versus the utility of~\cite{feldman2018streaming} increases.


\begin{figure}[ht!]
\centering
\includegraphics[width=120mm]{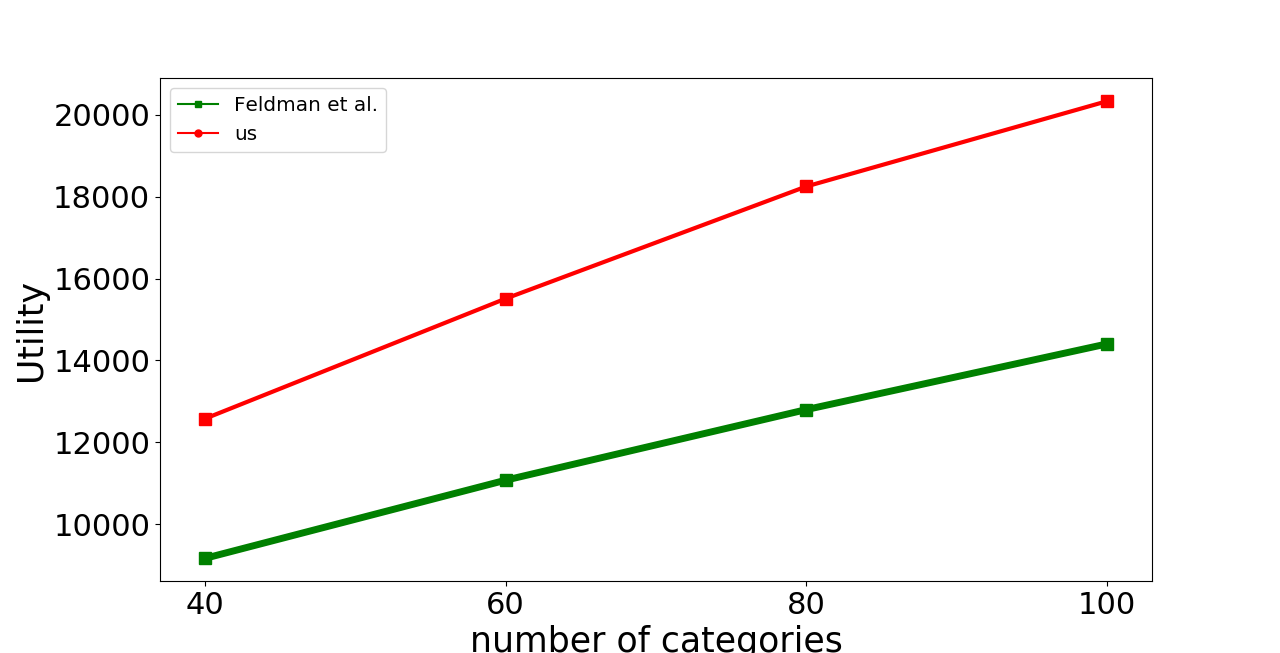}
\caption{ The plot is for 3-matchoid constraint, and $\alpha =3$, $\beta=2$}
\label{fig:p3}
\end{figure}

\begin{figure}[ht!]
\centering
\includegraphics[width=120mm]{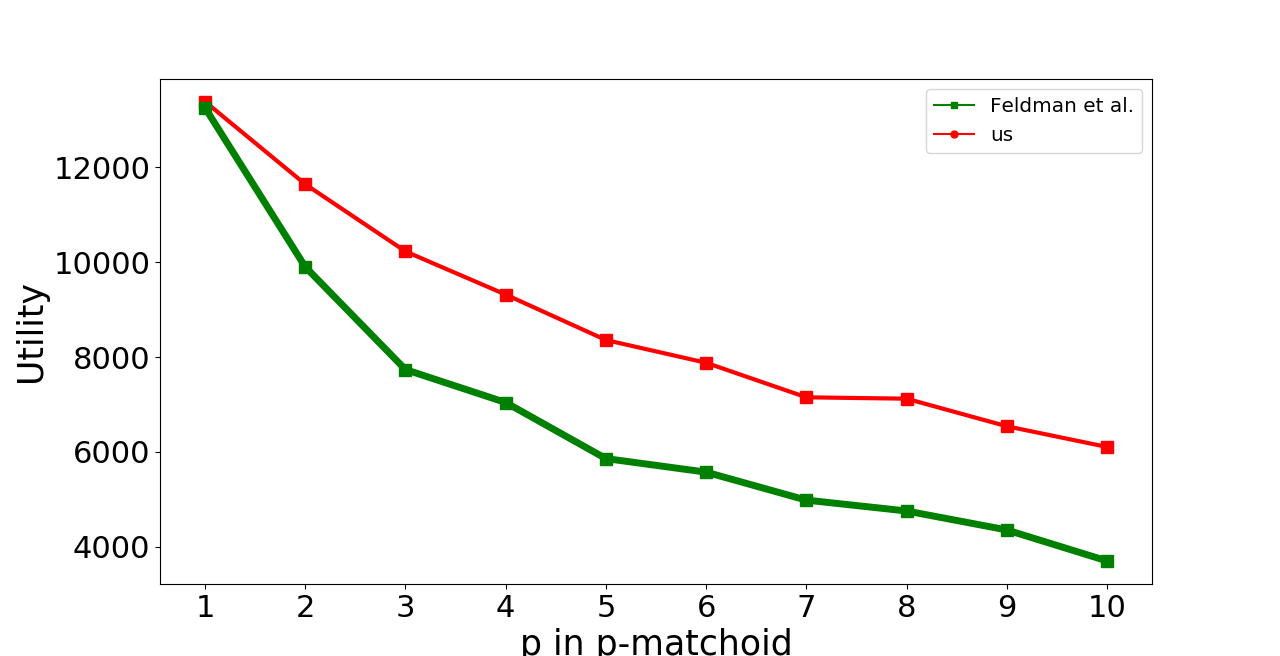}
\caption{ The plot is for $p$-matchoid constraint, for $p=1,\cdots, 10$, and $\alpha =3$, $\beta=2$ and fixed $k=30$.}
\label{fig:pmatchoid}
\end{figure}

 

 


\bibliographystyle{plainnat}
\bibliography{mybib}

\begin{thebibliography}{24}
\providecommand{\natexlab}[1]{#1}
\providecommand{\url}[1]{\texttt{#1}}
\expandafter\ifx\csname urlstyle\endcsname\relax
  \providecommand{\doi}[1]{doi: #1}\else
  \providecommand{\doi}{doi: \begingroup \urlstyle{rm}\Url}\fi

\bibitem[Abam et~al.(2014)Abam, Rezaei~Seraji, and Shadravan]{abam2014online}
MA~Abam, MJ~Rezaei~Seraji, and M~Shadravan.
\newblock Online conflict-free coloring of intervals.
\newblock \emph{Scientia Iranica}, 21\penalty0 (6):\penalty0 2138--2141, 2014.

\bibitem[Agrawal et~al.(2018)Agrawal, Shadravan, and Stein]{us}
Shipra Agrawal, Mohammad Shadravan, and Cliff Stein.
\newblock Submodular secretary problem with shortlists.
\newblock In \emph{10th Innovations in Theoretical Computer Science Conference
  (ITCS 2019)}. Schloss Dagstuhl-Leibniz-Zentrum fuer Informatik, 2018.

\bibitem[Babaioff et~al.(2008)Babaioff, Immorlica, Kempe, and
  Kleinberg]{Babaioff:2008}
Moshe Babaioff, Nicole Immorlica, David Kempe, and Robert Kleinberg.
\newblock Online auctions and generalized secretary problems.
\newblock \emph{SIGecom Exch.}, 7\penalty0 (2):\penalty0 7:1--7:11, June 2008.

\bibitem[Badanidiyuru et~al.(2014)Badanidiyuru, Mirzasoleiman, Karbasi, and
  Krause]{Badanidiyuru2014StreamingFly}
Ashwinkumar Badanidiyuru, Baharan Mirzasoleiman, Amin Karbasi, and Andreas
  Krause.
\newblock Streaming submodular maximization: Massive data summarization on the
  fly.
\newblock In \emph{Proceedings of the 20th ACM SIGKDD International Conference
  on Knowledge Discovery and Data Mining}, KDD '14, pages 671--680, New York,
  NY, USA, 2014. ACM.

\bibitem[Bateni et~al.(2013)Bateni, Hajiaghayi, and Zadimoghaddam]{Bateni}
Mohammadhossein Bateni, Mohammadtaghi Hajiaghayi, and Morteza Zadimoghaddam.
\newblock Submodular secretary problem and extensions.
\newblock \emph{ACM Trans. Algorithms}, 9\penalty0 (4):\penalty0 32:1--32:23,
  October 2013.

\bibitem[Buchbinder et~al.(2014)Buchbinder, Feldman, Naor, and
  Schwartz]{Buchbinder:2014}
Niv Buchbinder, Moran Feldman, Joseph~(Seffi) Naor, and Roy Schwartz.
\newblock Submodular maximization with cardinality constraints.
\newblock In \emph{Proceedings of the Twenty-fifth Annual ACM-SIAM Symposium on
  Discrete Algorithms}, SODA '14, pages 1433--1452, Philadelphia, PA, USA,
  2014. Society for Industrial and Applied Mathematics.

\bibitem[Chakrabarti and Kale(2015)]{Chakrabarti2015}
Amit Chakrabarti and Sagar Kale.
\newblock Submodular maximization meets streaming: matchings, matroids, and
  more.
\newblock \emph{Mathematical Programming}, 154\penalty0 (1):\penalty0 225--247,
  Dec 2015.

\bibitem[Chekuri et~al.(2015)Chekuri, Gupta, and Quanrud]{Chekuri}
Chandra Chekuri, Shalmoli Gupta, and Kent Quanrud.
\newblock Streaming algorithms for submodular function maximization.
\newblock In Magn{\'u}s~M. Halld{\'o}rsson, Kazuo Iwama, Naoki Kobayashi, and
  Bettina Speckmann, editors, \emph{Automata, Languages, and Programming},
  pages 318--330, Berlin, Heidelberg, 2015. Springer Berlin Heidelberg.
\newblock ISBN 978-3-662-47672-7.

\bibitem[Feige et~al.(2011)Feige, Mirrokni, and Vondr\'{a}k]{Feige:2011}
Uriel Feige, Vahab~S. Mirrokni, and Jan Vondr\'{a}k.
\newblock Maximizing non-monotone submodular functions.
\newblock \emph{SIAM J. Comput.}, 40\penalty0 (4):\penalty0 1133--1153, July
  2011.

\bibitem[Feldman and Zenklusen(2015)]{Feldman:2015}
Moran Feldman and Rico Zenklusen.
\newblock The submodular secretary problem goes linear.
\newblock In \emph{Proceedings of the 2015 IEEE 56th Annual Symposium on
  Foundations of Computer Science (FOCS)}, FOCS '15, pages 486--505,
  Washington, DC, USA, 2015. IEEE Computer Society.

\bibitem[Feldman et~al.(2014)Feldman, Svensson, and
  Zenklusen]{feldman2014simple}
Moran Feldman, Ola Svensson, and Rico Zenklusen.
\newblock A simple o (log log (rank))-competitive algorithm for the matroid
  secretary problem.
\newblock In \emph{Proceedings of the twenty-sixth annual ACM-SIAM symposium on
  Discrete algorithms}, pages 1189--1201. SIAM, 2014.

\bibitem[Feldman et~al.(2018)Feldman, Karbasi, and
  Kazemi]{feldman2018streaming}
Moran Feldman, Amin Karbasi, and Ehsan Kazemi.
\newblock Do less, get more: Streaming submodular maximization with
  subsampling, 2018.

\bibitem[Friggstad et~al.(2014)Friggstad, K{\"o}nemann, Kun-Ko, Louis,
  Shadravan, and Tulsiani]{friggstad2014linear}
Zachary Friggstad, Jochen K{\"o}nemann, Young Kun-Ko, Anand Louis, Mohammad
  Shadravan, and Madhur Tulsiani.
\newblock Linear programming hierarchies suffice for directed steiner tree.
\newblock In \emph{International Conference on Integer Programming and
  Combinatorial Optimization}, pages 285--296. Springer, 2014.

\bibitem[Friggstad et~al.(2016)Friggstad, K{\"o}nemann, and
  Shadravan]{friggstad2016logarithmic}
Zachary Friggstad, Jochen K{\"o}nemann, and Mohammad Shadravan.
\newblock {A Logarithmic Integrality Gap Bound for Directed Steiner Tree in
  Quasi-bipartite Graphs }.
\newblock In \emph{15th Scandinavian Symposium and Workshops on Algorithm
  Theory (SWAT 2016)}, Leibniz International Proceedings in Informatics
  (LIPIcs), pages 3:1--3:11, 2016.

\bibitem[Ghuge and Nagarajan()]{ghugequasi}
Rohan Ghuge and Viswanath Nagarajan.
\newblock Quasi-polynomial algorithms for submodular tree orienteering and
  other directed network design problems.
\newblock In \emph{Proceedings of the Fourteenth Annual ACM-SIAM Symposium on
  Discrete Algorithms}, pages 1039--1048. SIAM.

\bibitem[Gijswijt et~al.(2007)Gijswijt, Jost, and
  Queyranne]{gijswijt2007clique}
Dion Gijswijt, Vincent Jost, and Maurice Queyranne.
\newblock Clique partitioning of interval graphs with submodular costs on the
  cliques.
\newblock \emph{RAIRO-Operations Research}, 41\penalty0 (3):\penalty0 275--287,
  2007.

\bibitem[Gupta et~al.(2010)Gupta, Roth, Schoenebeck, and Talwar]{Gupta:2010}
Anupam Gupta, Aaron Roth, Grant Schoenebeck, and Kunal Talwar.
\newblock Constrained non-monotone submodular maximization: Offline and
  secretary algorithms.
\newblock In \emph{Proceedings of the 6th International Conference on Internet
  and Network Economics}, WINE'10, pages 246--257, Berlin, Heidelberg, 2010.
  Springer-Verlag.

\bibitem[Hazan et~al.(2006)Hazan, Safra, and Schwartz]{hazan2006complexity}
Elad Hazan, Shmuel Safra, and Oded Schwartz.
\newblock On the complexity of approximating k-set packing.
\newblock \emph{computational complexity}, 15\penalty0 (1):\penalty0 20--39,
  2006.

\bibitem[Kazemi et~al.(2019)Kazemi, Mitrovic, Zadimoghaddam, Lattanzi, and
  Karbasi]{kazemi2019submodular}
Ehsan Kazemi, Marko Mitrovic, Morteza Zadimoghaddam, Silvio Lattanzi, and Amin
  Karbasi.
\newblock Submodular streaming in all its glory: Tight approximation, minimum
  memory and low adaptive complexity, 2019.

\bibitem[Kesselheim and T{\"o}nnis(2017)]{kesselheim}
Thomas Kesselheim and Andreas T{\"o}nnis.
\newblock {Submodular Secretary Problems: Cardinality, Matching, and Linear
  Constraints}.
\newblock In \emph{Approximation, Randomization, and Combinatorial
  Optimization. Algorithms and Techniques (APPROX/RANDOM 2017)}, Leibniz
  International Proceedings in Informatics (LIPIcs), pages 16:1--16:22, 2017.

\bibitem[Lachish(2014)]{Lachish14}
Oded Lachish.
\newblock O(log log rank) competitive-ratio for the matroid secretary problem.
\newblock \emph{CoRR}, abs/1403.7343, 2014.
\newblock URL \url{http://arxiv.org/abs/1403.7343}.

\bibitem[Nemhauser et~al.(1978)Nemhauser, Wolsey, and
  Fisher]{nemhauser1978analysis}
George~L Nemhauser, Laurence~A Wolsey, and Marshall~L Fisher.
\newblock An analysis of approximations for maximizing submodular set
  functions—i.
\newblock \emph{Mathematical programming}, 14\penalty0 (1):\penalty0 265--294,
  1978.

\bibitem[Norouzi-Fard et~al.(2018)Norouzi-Fard, Tarnawski, Mitrovic, Zandieh,
  Mousavifar, and Svensson]{norouzi}
Ashkan Norouzi-Fard, Jakub Tarnawski, Slobodan Mitrovic, Amir Zandieh,
  Aidasadat Mousavifar, and Ola Svensson.
\newblock Beyond 1/2-approximation for submodular maximization on massive data
  streams.
\newblock In \emph{Proceedings of the 35th International Conference on Machine
  Learning}, volume~80, pages 3829--3838. PMLR, 10--15 Jul 2018.

\bibitem[Soto(2013)]{soto2013matroid}
Jos{\'e}~A Soto.
\newblock Matroid secretary problem in the random-assignment model.
\newblock \emph{SIAM Journal on Computing}, 42\penalty0 (1):\penalty0 178--211,
  2013.

\end{thebibliography}


\end{document}